\documentclass[preprint,10pt]{elsarticle}
\usepackage[tight,footnotesize]{subfigure}
\usepackage{graphicx}
\DeclareGraphicsExtensions{.eps,.ps}
\usepackage[tight,footnotesize]{subfigure}
\usepackage{epsf}    
\usepackage{amsmath}
\usepackage{amssymb}    
\usepackage{latexsym} 
\usepackage{algorithm}
\usepackage{algorithmic}
\makeatletter
\newcommand\fs@norules{\def\@fs@cfont{\bfseries}\let\@fs@capt\floatc@ruled
  \def\@fs@pre{}%
  \def\@fs@post{}%
  \def\@fs@mid{\kern3pt}%
  \let\@fs@iftopcapt\iftrue}
\makeatother
\floatstyle{norules}
\restylefloat{algorithm}
\usepackage{tabularx, booktabs}
\newcolumntype{Y}{>{\centering\arraybackslash}X}
     
\newtheorem{theoremm}{Theorem}   
\newtheorem{eqed}{Example}     
\newtheorem {lemmaa}{Lemma}

\newtheorem {defnn}{Definition} 
\newtheorem {corollaryy}{Corollary}
\newtheorem {conjecturee}[theoremm]{Conjecture}
     
\newtheorem {procd}{Procedure}   

\newenvironment{defn}{\begin{defnn} \sl}{\end{defnn}}  
\newenvironment{proof}{\noindent {\bf Proof :\ } }{\hfill$\Box$ }

\biboptions{comma,square,sort&compress}

\journal{Physica D}
\begin{document}

\begin{frontmatter}
\title{Asynchronous Cellular Automata and \\Pattern Classification \thanks{This work is supported by DST Fast Track Project Fund (No.SR/FTP/ETA-0071/2011)}}   

\author[rvt]{Biswanath Sethi\corref{cor}}
\ead{sethi.biswanath@gmail.com}
\author[focal]{Souvik Roy}
\ead{svkr89@gmail.com}
\author[focal]{Sukanta Das}
\ead{sukanta@it.iiests.ac.in}

\cortext[cor]{Corresponding author}
\address[rvt]{Department of Computer Science Engineering \& Applications\\ Indira Gandhi Institute of Technology, Sarang\\ Odisha, India-759146}
\address[focal]{Department of Information Technology\\ Indian Institute of Engineering Science \& Technology, Shibpur\\ West Bengal, India-711103}

 

\begin{abstract}  
This paper designs an efficient two-class pattern classifier utilizing asynchronous cellular automata (ACAs). The two-state three-neighborhood one-dimensional ACAs that converge to fixed points from arbitrary seeds are used here for pattern classification. To design the classifier, we first identify a set of ACAs that always converge to fixed points from any seeds with following properties - (1) each ACA should have at least two but not huge number of fixed point attractors, and (2) the convergence time of these ACAs are not to be exponential. In order to address the first issue, we propose a graph, coined as {\em fixed point graph} of an ACA that facilitates in counting the fixed points. We further perform an experimental study to estimate the convergence time of ACAs, and find that there are some convergent ACAs which demand exponential convergence time. Finally, we find that there are 71 (out of 256) ACAs which can be effective candidates as pattern classifier.  We use each of the candidate ACAs on some standard data sets, and observe the effectiveness of each ACAs as pattern classifier. It is observed that the proposed classifier is very competitive and performs reliably better than many standard existing algorithms.
\end{abstract}  
\end{frontmatter}

\section{Introduction}   
\label{intro}  
Cellular automata (CAs) were introduced by Jon von Neumann in 1950s to primarily model biological self-reproduction \cite{Neuma66, Wolframbook1, LangtonIII, langton86, Burks70, chris2004, reggia1998}. However, these systems soon captured the attention of researchers due to their massive parallelism and ability of modelling physical systems \cite{Chopard, Codd68, palash3, Maxgarzon, mitchell93}. They are also proved to be computationally universal and more expressive than turing machine \cite{toffoli90, Thatcher, wolfram84b, jarkko2005, gutowitz91, cook2004}. A cellular automaton (CA) is defined by a lattice of cells and a local rule. The system evolves in discrete time and space, and all the cells follow same rule to generate its next state. The von Neumann's CAs were two dimensional with 29 states per cell and each cell is dependent on itself and its four non-orthogonal neighbors. The CA structure was later simplified by many researchers, and finally a two-state three-neighborhood CA structure was proposed on one-dimensional lattice \cite{Wolfr83, wolfram84b}. Since 1980s, the researchers have been showing that even these simple CAs can model complex systems \cite{Chopard, wolfram86, wolfram1984, gianluca1995, jarkko2005, gorodkin1993, hoekstra2010}. Such CA can find wide range of interesting applications such as modelling some natural growth processes like seashell patterns and snowflakes \cite{Meinhardt1995, patel2015}. 

However, investigations and studies on CAs are mostly centred around synchronous (deterministic) CAs where all the cells are updated simultaneously. In other words, like other synchronous systems, the CA assumes a global clock that forces the cells to get updated simultaneously. Though the synchronous CAs are proved to be good in modelling physical systems, the assumption of global clock is not very natural. On the other hand, in {\em asynchronous} CAs (ACAs), cells are independent and are, therefore, updated independently during the evolution of the system. Hence, the choice of ACAs with independent cell dynamics as models of physical systems are better \cite{Fates14, damien2009}. In recent times, there is a growing interest on ACAs since they can better model the natural systems \cite{bersinietal94, Inger84, Fates14, Nehaniv2003, Maji1, patel2015}. Researchers are more concerned in the properties of self-replications on ACAs \cite{ Huang5, Huang4, Huang3, Huang1, Ytakada}. Self-replication on ACAs allows avoidance of defective parts and simplifies programming of computers and therefore can be used in nano computers, where reconfigurability is an essential property. Also, the independently updated scheme of ACAs may be appropriate for modelling social networks or computer network communications. Although the applications of ACAs are limited, in a recent work, Bolt et al. \cite{Bolt2015}, from a theoretical as well as practical analysis on stochastic cellular automata, have provided a formal description of their properties suitable for applications in the domain of systems modelling without the need of a strong mathematical background.

The design of pattern classifiers using CAs are widely reported in literature where synchronous CAs have been used \cite{DasMNS09, das:modeling, NiloyPAMI, NiloyI, Maji1}. It should be mentioned here that, the issue of pattern classification was not previously tackled with ACAs. Therefore, the objective of this work is to design a two-class pattern classifier using ACAs. Here, for an ACA, we adopt one-dimensional two-state three-neighborhood CAs that are updated asynchronously.  In order to achieve the goal, we first characterized the ACAs. We observed that during their evolution, some of the ACAs converged to fixed points from an arbitrary seed. A set of configurations/states of an ACA that approach to a fixed point can be considered as the patterns of a single class. This idea is explored in developing real-life pattern classifier. A preliminary work on this is already reported in \cite{hpcs}.

We report the design of two-class pattern classifier in Section \ref{design}. Since we utilize ACA that always converges to a fixed point in the design of pattern classifier, our first task is to identify such ACAs. In two-state three-neighborhood one-dimensional system, there are 256 local rules \cite{wolfram86}. Convergence of a set of 64 (out of 256) rules under asynchronous update has been studied in \cite{nazim}. We explore here all of the 256 rules, and identify a set of ACAs (see Table \ref{AFP}) as convergent ACAs (Section \ref{converged}). However, a convergent ACA having a single fixed point can not act as two-class pattern classifier. Similarly, a convergent ACA having huge number of fixed points is not a good classifier. We, therefore, develop an algorithm to count fixed points of an $n$-cell ACA, and then decide which ACAs can and cannot act as good classifier (Section \ref{mulFP}).

It is further observed that some convergent ACAs take huge amount of time to converge. It is not practical to use these ACAs as real-life pattern classifier. In this work, we develop a method to estimate average convergence time of convergent ACAs. As a next step, we point out the convergent ACAs that are having exponential (average) convergence time, and exclude them from the list of candidate classifiers (Section\ref{ele}). Finally, we use each of the candidate ACAs on some standard data sets, and observe the effectiveness of each ACAs as pattern classifier (Section \ref{per}). The ACA that shows highest {\em efficiency} against a given data set, is considered as a classifier for that data set. We observe that the proposed classifier is very competitive and performs reliably better than many standard existing algorithms (see Table \ref{compdata}). Before going to address the issue we intend to present in brief the preliminaries of CAs in the next section.

\section{Cellular Automata Preliminaries}   
\label{dfe} 
       
A CA is a discrete dynamical system which evolves in discrete space and time. It consists of a lattice of cells, each of which stores a variable at time $t$, that refers to the present state of the CA cell \cite{Neuma66}. In this work, we consider, one-dimensional three-neighborhood binary CAs with periodic boundary condition, where the cells are arranged as a ring. The next state of each CA cell is determined as
\begin{equation}    
S^{t+1}_i = f({S^t_{i-1}},{S^t_i},{S^t_{i+1}})   
\end{equation}   
where $f$ is the next state function, ${S^t_{i-1}}$, ${S^t_i}$ and ${S^t_{i+1}}$ are the present states of the left neighbor, self and right neighbor of the $i^{th}$ CA cell at time $t$ respectively. A collection of (local) states $\mathcal{S}^t(S_1^t, S_2^t,\cdots,S_n^t)$ of cells at time $t$ is referred to as a configuration or a (global) state of the CA at $t$. The function $f:\{0,1\}^3\mapsto\{0,1\}$ can be expressed as a look-up table (see Table \ref{Trules}). The decimal equivalent of the 8 outputs is called `rule' \cite{wolfram86}. There are $2^8=256$ CA rules in two-state three-neighborhood dependency. Three such rules (40, 99 and 219) are shown in Table \ref{Trules}. First row of the table shows the possible combinations of present states (PSs) of left, self and right neighbors of a cell. Whereas, third, fourth and fifth rows show the next states (NSs) of corresponding PSs. The last column notes the rules.
  
\begin{defnn}               
The association of the neighborhood $ x,y,z $ to the value $ f(x,y,z)$, which represents the result of the updating function, is called {\em Rule Min Term} (RMT). Each RMT is associated to a number $ r(x,y,z)= 4x+2y+z $. 
\end{defnn}
The first row of Table \ref{Trules} shows the 8 possible RMTs of three-neighborhood CA.

A CA state can be viewed as a sequence of RMTs. For example, the state 1110 in periodic boundary condition can be viewed as $\langle3765\rangle$, where 3, 7, 6 and 5 are corresponding RMTs on which the transitions of first, second, third and fourth cell can be made. To get a sequence of RMTs for a state, we consider an imaginary 3-bit window that slides over the state. The window contains a 3-bit binary value which is equivalent to an RMT. To get the $i^{th}$ RMT, the window is loaded with $(i-1)^{th}$, $i^{th}$ and $(i+1)^{th}$ bits of the state. The window slides one bit right to report the $(i+1)^{th}$ RMT. Now the current content of the window is $i^{th}$, $(i+1)^{th}$ and $(i+2)^{th}$ bit of the state. In the sequence of RMTs, however, two consecutive RMTs are related. If 5 (101) is the $i^{th}$ RMT in some sequence, then $(i+1)^{th}$ RMT is either 2 (010) or 3 (011). Similarly, if 0 (000) or 4 (100) is the $i^{th}$ RMT, then 0 (000) or 1 (001) is the $(i+1)^{th}$ RMT. The relations of two consecutive RMTs in a sequence of RMTs are noted in Table \ref{nextrmt}. 

\begin{defn}
An RMT $r (x,y,z)$ of a rule is active if $f(x,y,z)\neq y $ and otherwise passive.
\end{defn}  
          
In rule 219, RMT 1 (001) is active and RMT 6 (110) is passive (see Table \ref{Trules}).
\begin{eqed}
As a proof of concept, consider the evolution of rule 219 ACA with only 4 cells. Assume the initial state is 0101 (Figure \ref{ACAtransition}(b)). Selecting the second cell to update (shown in figure \ref{ACAtransition}(b)) the next state is 0001 as RMT 2 is active for rule 219 ACA. Updating the 3rd cell in state 0001, we get the state 0011 and updating the 2nd cell again in state 0011, we reach at state 0111. Updating any cell in state 0111, the ACA remains in state 0111 forever, since RMTs 5, 3, 7 and 6 are passive for rule 219 ACA (see Table \ref{Trules}). 
\end{eqed}
 
\begin{table}  
\caption{Look-up table for rules 40, 99 and 219 }           
\centering               
\label{Trules}      
$\begin{array}{cccccccccc}               
{\rm PSs:}&111 & 110 & 101 & 100 & 011 &  010 &  001 & 000 & Rule  \\                  
(RMT) & (7) & (6) & (5) & (4) & (3) & (2) & (1) & (0) & \\                  
{\rm(i)NSs:} & 0 & 0&1 &0 & 1 & 0 & 0 & 0 & 40 \\        
{\rm(ii)NSs:} & 0 & 1&1 & 0 & 0 &0 & 1 &1& 99 \\                  
{\rm (iii)NSs:}& 1 & 1 & 0 &1 & 1 &0 & 1& 1& 219 \\   
\end{array}$                
\end{table} 

\begin{table}   
\caption{Relationship between $i^{th}$ and $(i+1)^{th}$ RMTs}   
\label{nextrmt}   
\centering 
\begin{tabular}{|c|c|}\hline 
$i^{th}$ RMT&$(i+1)^{th}$ RMT\\ \hline  
0 &0, 1 \\   
1 &2, 3 \\    
2 &4, 5 \\   
3 &6, 7 \\  
4 & 0, 1 \\
5  & 2, 3 \\
6 & 4, 5 \\
7 & 6, 7 \\
\hline   
\end{tabular}    
\end{table}  
 
Traditionally, all the cells of a CA are forced to get updated simultaneously. This constraint is relaxed in an asynchronous CA, where the cells can act independently. Though asynchronism is considered as an uncontrolled phenomenon, it is generally modelled as a stochastic process. {\em Fully asynchronous update}, where an arbitrary cell is updated at each step is one primary scheme to evolve an ACA \cite{nazim}. In this paper, we evolve ACAs under fully asynchronous update. Figure \ref{ACAtransition} shows the partial state transition diagram of a 4-cell rule 219 ACA. The cells, updated fully asynchronously during state transitions, are noted beside arrows. Note that each state of figure \ref{ACAtransition} converges to a {\em fixed point}.

During the evolution of an ACA, a sequence $(u_t)_{t\in {\mathbb{N}}}$ of cells can be observed where $u_t$ denotes the cell updated at time $t$. We call the sequence as {\em update pattern} \cite{tamc2014}. For an initial condition $ x $ and an update pattern $U$, the evolution of the system is given by the sequence of states $ (x^t) $ obtained by successive applications of the updates of $ U $. Formally, we have:
$ x^{t+1} = F(x^t, u_t) $ and $ x^0 = x $, with:
$$ x^{t+1}_i = 
\begin{cases}
f (x_{i-1}^t,x_i^t, x_{i+1}^t)&  \text{ if }  i = u_t \\
x^t_i  & \text{ otherwise.}
\end{cases}$$
This evolution can be represented in the form of {\em a state transition diagram}. For $x = 0000$ and $U = (1, 4, 3, 1, \ldots)$, the 4-cell rule 219 ACA converges to a {\em fixed point} (1011), which is shown in figure \ref{ACAtransition}(d).

\begin{defnn}    
A fixed point is an ACA state, next state of which is the state itself for any random update of cells. That is, if an ACA reaches to a fixed point, the ACA remains in that particular state forever.  
\end{defnn}  
 
\begin{figure}   
\centering     
\includegraphics[scale=0.5]{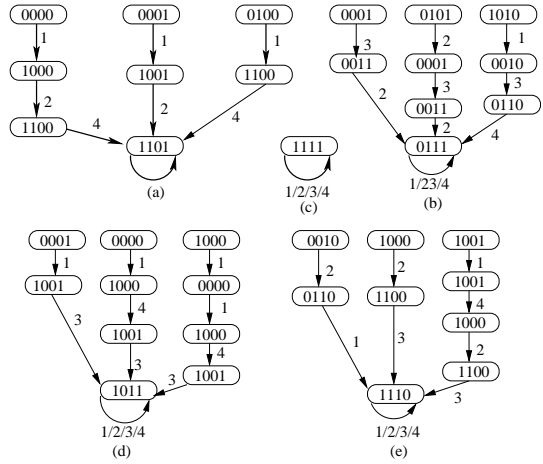}    
\caption{Partial state transition diagram for 4-cell rule 219 ACA.}  
\label{ACAtransition}    
\end{figure} 

In figure \ref{ACAtransition}, the states 1101, 0111, 1011, 1110 and 1111 are fixed points. The next state of 1111 is always 1111 for any random update of cells, because RMT 7 of rule 219 is passive. It can be observed that all the RMTs of a fixed point are passive. Hence, the following lemma can be obtained.  

\begin{lemmaa}   
\label{lmf}  
Rule $R$ ACA forms a fixed point with state $\mathcal{S}$ if all the RMTs of the RMT sequence of $\mathcal{S}$ are passive.  
\end{lemmaa}

Since there are five attractors in 4-cell rule 219 ACA (Figure \ref{ACAtransition}), we can find five basins of attraction. A fixed point is the representative of the corresponding attractor basin. 

\section{Design of pattern classifier}       
\label{design}  
An {\em n}-cell ACA with multiple fixed points can act as natural classifier. Each class contains a set of states that converge to a fixed point. To identify the class of patterns, the fixed points, representing the classes, need to be stored in memory (Figure \ref{archiACA}). For the identification of the class of an input pattern {\em p}, the ACA is loaded with {\em p} and continuously updated till it reaches to any fixed point. Then, from the fixed point and the stored information, one can declare the class of the input pattern {\em p}. In figure \ref{archiACA}, the class of {\em p} is I. However, if there are more than two fixed point attractors, then a set of fixed points identify the class.  
 
\begin{figure}[h]   
\centering      
\includegraphics[height=1.6in,width=3in]{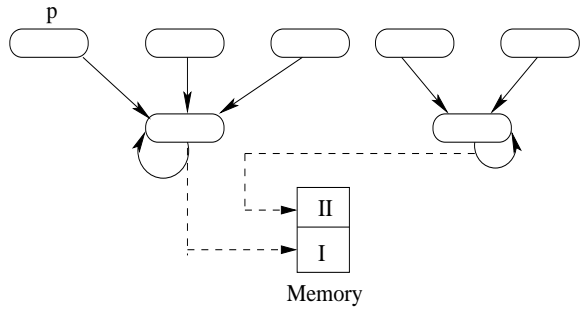}      
\caption{Multiple fixed point based classification strategy}   
\label{archiACA}   
\end{figure}  

\begin{eqed} 
\label{pe}
Let us show that the ACA of figure \ref{ACAtransition} can act as a two class classifier. Assume that the fixed points 0111, 1111 and 1110 are of class 1 and the rest 1011 and 1101 are of class 2.  Hence, the ACA of figure \ref{ACAtransition} can act as a two-class pattern classifier.  
\end{eqed} 

A proper distribution of the patterns among the CA attractor basins is necessary to design a CA based two-class classifier. This design theoretically requires a proper distribution of CA fixed points among the two pattern sets {\em $P_1$} and {\em $P_2$}. However, for real-life data-sets, the attractor basins may mix up the patterns of two classes. Therefore, the primary metric for evaluating the performance of classifier is the classification accuracy. It is measured as:\\
\begin{equation} 
\label{equat}  
Efficiency=\frac{\# ~patterns~ properly~ classified}{Total~ No.~ of~ patterns}\times 100~\%
\end{equation}   
 
\begin{eqed}   
Let us consider the ACA of figure \ref{ACAtransition} as a two-class classifier (see Example \ref{pe}). Assuming {\em $P_1$}=\{1111, 0011,1100, 0110\} and {\em $P_2$}=\{0000, 0001,1000\} are the two pattern sets. We now consider {\em $P_1$} as patterns of class 1, and {\em $P_2$} as patterns of class 2. However, the pattern 1000 is of class 2, but suppose it is wrongly identified by the classifier as in class 1. Therefore, out of 7 patterns only 6 patterns are properly identified. So the classification efficiency is 85.714\% for the above pattern set and the classifier. We can vary the efficiency by changing the pattern set and the ACA.
\end{eqed} 

The above discussion shows that ACAs that converge to fixed points during their evolution can act as a pattern classifiers. However, an arbitrary ACA can not be a two-class classifier. The following tasks are, therefore, identified to get appropriate candidates that can act as real-life pattern classifier.
\begin{enumerate}
\item All of the 256 ACAs do not converge to fixed points. So the first task is to identify the ACAs that converge to some fixed points from an arbitrary seed.
\item An ACA of the above set may have a single fixed point, to which all possible ACA states converge. Such an ACA can not act as two-class classifier. Similarly, an ACA with a huge number of fixed points can not be a good classifier. Since we store the fixed points, memory overhead is increased in that case. So, our next task is to find a subset of above set of ACAs that has atleast two but not a huge number of fixed points.
\item Convergence time of an ACA may be huge which can make the ACA inappropriate to be a practical pattern classifier. Our final task is, therefore, to exclude those ACAs from the above subset which requires exponential time to converge to a fixed point.
\end{enumerate}
The subsequent sections (Sections \ref{converged}, \ref{mulFP} and \ref{ele}) handle these issues to get a set of candidate ACAs as pattern classifiers. The performance of the proposed classifier is evaluated utilizing these candidate ACAs.

\section{ Identification of convergent ACAs }  
\label{converged} 
This section identifies those ACAs, which always converge to some fixed point attractors during their evolution. Following theorem states the condition of ACAs to be convergent ACAs.
 
\begin{theoremm}     
\label{fixed}      
Rule $R$ ACA converges to fixed point attractor if one of the following condition is satisfied:\\ (i) RMT 0 (resp. RMT 7) of $R$ is passive and RMT 2 (resp. RMT 5) is active.\\(ii) RMTs 0, 1, 2 and 4 (resp. RMTs 3, 5, 6 and 7) are passive and RMT 3 or 6 (resp. RMT 1 or 4 ) is active.\\ (iii) RMTs 1, 2, 4 and 5 (resp. RMTs 2, 3, 5 and  6) are passive.
\end{theoremm} 

\begin{proof}        
\underline{Proof of case (i)}: Let us consider RMT 0 of $R$ is passive and RMT 2 is active. We shall show that the rule $R$ ACA can reach to a fixed point attractor from any initial state. Since RMT 0 is passive, the all-0 state (RMT sequence $\langle 00\cdots0\rangle$) is a fixed point attractor. In any other state (except all-1), a sequence of consecutive 1s guided by 0s can always be found.  
 
Consider, RMT 7 of $R$ is active. Now in such a state (like $\cdots0111110\cdots$), we can find RMT 7 in its corresponding RMT sequence. If a cell with RMT 7 is selected to update, then, the sequence of consecutive 1s is divided into two sub-sequences of consecutive 1s guided by 0s. The new sequences have less number of 1s (like $\cdots0111110\cdots\rightarrow\cdots0101110\cdots$). After a number of similar updates, we can get a state with a number of single 1 and two consecutive 1s guided by 0s. A cell has state 1 with left and right neighbor's states as 0s (010) implies that the cell can act on RMT 2. Since RMT 2 is active, all such cells can reach to state 0. So, finally, we get either the all-0 state (that is, the ACA is converged to all-0 state), or a state with two consecutive 1s  guided by 0s $(\cdots001100\cdots)$. In second case, the RMT sequence contains RMT 0 and RMTs 1, 3, 4 and 6. If RMTs 1, 3, 4 and 6 are passive, the state $\cdots0110\cdots$ itself is a fixed point attractor. Otherwise, the ACA can reach to all-0 (fixed point) state  after some update of cells. Hence, from an arbitrary state with sequences of consecutive 1s guided by 0s, the ACA can reach to a fixed point attractor. The all-1 state is a special state which contains no 0. However, if an arbitrary cell is updated, then we can get a 0, and then the new state can reach to fixed point attractor with the above logic.  

Now consider, RMT 7 is passive. Then, the all-1 state (RMT sequence $\langle77\cdots7\rangle$) is another fixed point attractor. A state with a sequence of consecutive 1s, guided by 0s, contains RMTs 3 and 6. If any one of them is active, the ACA with that state can reach to all-0 fixed point. If both RMTs (RMT 3 and 6) are passive but RMT 1, 4 or 5 is active, the ACA can reach to all-1 fixed point. If all RMTs except RMT 2 are passive, then the state itself is a fixed point attractor. Hence, the rule $R$ ACA that can reach to fixed point attractors from any initial state if RMT 0 of $R$ is passive and RMT 2 is active.  

While RMT 7 is passive and RMT 5 is active (and the rest RMTs are active or passive) it can be shown by similar logic that the rule $R$ ACA converge to fixed point attractors \\ 
  
\underline{Proof of case (ii)}: Consider RMTs 0, 1, 2 and 4 of $R$ are passive. Then, the states where two 1s are separated by at least two consecutive 0s (like $\cdots001000100\cdots$) are fixed point attractors, because the corresponding RMT sequences of these states contain only RMTs 0, 1, 2 and 4 (Lemma \ref{lmf}). Now consider a state which contains two or more consecutive 1s. Then, the corresponding RMT sequence of the state contains RMTs 3, 6 and 7 (along with other RMTs). If RMT 3 or 6 is active, the number of 1s can be reduced to a single 1 separated by 0s during evolution  of the ACA. The resultant state is a fixed point attractor if corresponding RMT sequence contains RMTs 0, 1, 2 and 4 only. If the resultant state (like  $\cdots001010\cdots$) contains any other RMT (RMT 5 in this case) which is active (except 3 or 6) then updating the cell with that active RMT, the state will have three consecutive 1s separated by 0s ($\cdots001110\cdots$). The RMT sequence of this state is now with RMTs 0, 1, 3, 7, 6 and 4. Now, updating the cell with RMT 3 or 6, the state reaches to a state with RMTs 0, 1, 2 and 4 only. Since RMTs 0, 1, 2 and 4 are passive, the state is a fixed point.
 
While RMTs 3, 5, 6, and 7 are passive, the states where two 0s are separated by at least two consecutive 1s (like $\cdots11011011\cdots$) are fixed point attractors because the corresponding RMT sequences of these states contain RMTs 3, 5, 6 and 7 only (Lemma \ref{lmf}). Now consider a state which contains two or more consecutive 0s. Therefore, the  corresponding RMT sequence of state contains RMTs 4, 0 and 1 (along with other). If RMT 1 or 4 is active, the number of 0s can be reduced to single 0 separated by 1s during the evolution of the ACA. The resultant state is a fixed point attractor if corresponding RMT sequence contains RMT 3, 5, 6, 7 only. If the resultant state contains any other RMT which is active (except RMT 1 or 4) then updating the cell with that active RMT, we can reach to a fixed point attractor.\\
  
\underline{Proof of case (iii)}: Consider  RMTs 1, 2, 4 and 5 of $R$ are passive. We  next show that from any initial state, the ACA can reach to a fixed point attractor. Let us consider RMTs 0 and 7 are active. Now, from all-0 state, updating cell with RMT 0 we can reach to the state $0101\cdots$ after number of updates. The state $0101\cdots$ is now with RMTs 2 and 5 only. This state is a fixed point attractor since RMT 2 and 5 are passive. The transitions are :~ $0000\cdots\rightarrow0100\cdots\rightarrow\cdots\rightarrow0101\cdots$. 

Similarly, from all-1 state, the ACA can reach to $0101\cdots01$ or $0101\cdots011$ (depending on the number of cells). The state $0101\cdots01$ is itself a fixed point attractor as RMTs 2 and 5 are passive and $0101\cdots011$ can be a fixed point attractor if RMTs 3 and 6 are passive. If RMT 3 or 6 is active, the ACA can obviously reach to a fixed point attractors by updating the cell with RMT 3 or 6. The transitions are: $0101\cdots011\rightarrow 0101\cdots001$. Now, it can easily be shown that the ACA can reach to a fixed point attractor from any state.  

Lastly, it can also be shown that rule $R$ ACA converges to fixed point attractors, if RMTs 2, 3, 5 and 6 are passive. We omit the detail steps of this proof because the rationale is similar with other cases.  
\end{proof}    

\begin{eqed}            
Let us consider the rule 40 ACA, in which RMTs 0, 1, 3 and 4 are passive (Table \ref{Trules}). Therefore, the rule satisfies the condition for convergence to fixed points (Theorem \ref{fixed} (i)). Here, we assume number of cells are 4 and the initial state is 1111. The ACA can reach to a fixed point attractor after some random update of cells from the initial state. One possible transition is: $1111 (2)\rightarrow1011 (1)\rightarrow0011 (4)\rightarrow0010 (3)\rightarrow0000$ (the cell updated in a step is noted in bracket).  
\end{eqed} 

There are 64 rules where the RMT 0 is passive and RMT 2 is active (Theorem \ref{fixed}(i)). 8 more rules can be identified where RMTs 0, 1, 2 and 4 are passive and RMT 3 is active (Theorem \ref{fixed}(ii)). In this way, we can get 146 ACAs (out of 256), which always approach to some fixed-point attractors. Such ACA rules are listed in Table \ref{AFP}. 
 
\begin{table}   
\caption{ACAs converge to fixed points}  
\label{AFP}   
\centering
\begin{tabular}{|cccccccccc|} \hline   
0&2&4&5&8&10&12&13&16&18  \\
24&26&32&34&36&40&42&44&48&50 \\
56&58&64&66&68&69&72&74&76&77 \\
78&79&80&82&88&90&92&93&94&95 \\
96&98&100&104&106&112&114&120&122&128\\
130&132&133&136&138&140&141&144&146&152 \\
154&160&161&162&163&164&165&166&167&168 \\
169&170&171&172&173&174&175&176&177&178 \\
179&180&181&182&183&184&185&186&187&188 \\
189&190&191&192&194&196&197&200&202&203 \\
204&205&206&207&208&210&216&217&218&219 \\
220&221&222&223&224&225&226&227&228&229 \\
230&231&232&233&234&235&236&237&238&239 \\
240&241&242&243&244&245&246&247&248&249\\
250&251&252&253&254&255&&&& \\  \hline 
\end{tabular}  
\end{table} 

\begin{corollaryy}
\label{cor}
An arbitrary state of an ACA, with multiple fixed point attractors, may converge to different fixed point attractors for different {\em update patterns}.
\end{corollaryy}

\begin{proof}
From the concept of the proof of the Theorem \ref{fixed}, it has been proved that any arbitrary state of an ACA can converge to fixed point attractor depending upon the active and passive RMTs of that particular ACA. It has also been shown that an ACA state can converge to fixed point, by selecting cells with active RMTs (Theorem \ref{fixed} (case-1)). Since ACA cells are updated randomly, so an ACA may converge to different fixed points following different sequence of update of cells  for different run. Hence, ACA with multiple fixed points, may converge to different fixed point attractors.
\end{proof}

\begin{eqed}
This example illustrates the convergence of a state of an ACA to two different fixed points following two different update sequences of cells. Let us consider an ACA with RMT 0 and RMT 7 as passive and RMT 2 and RMT 5 as active. Since both the RMTs 0 and 7 are passive for the ACA, the all-0 state (with RMT 0 only) and all-1 state (with RMT 7 only) are fixed points. We will now show a state can converge to both all-0 and all-1 fixed points following two different sequences of updates of cells. Consider the state 101010  as the seed for 6-cell ACA. The state 101010 can converge to all-0 state, updating 1$^{st}$, 3$^{rd}$ and 5$^{th}$ cell. Similarly, the state 101010 can also converge to all-1 state updating 2$^{nd}$, 4$^{th}$ and 6$^{th}$ cell. Here, we get two different update patterns. These are (1, 3, 5) and (2, 4, 6) respectively. The detail transitions are as follows:
\begin{itemize}
\item 101010 (1) $\rightarrow$ 001010 (3) $\rightarrow$ 000010 (5) $\rightarrow$ 000000.
\item 101010 (2) $\rightarrow$ 111010 (4) $\rightarrow$ 111110 (6) $\rightarrow$ 111111.
\end{itemize}
The cells updated during transitions are noted in bracket right to the state.
\end{eqed}

\section{Counting of fixed points}
\label{mulFP}

This section reports a scheme of counting fixed points of an ACA, and then excludes the ACAs from Table \ref{AFP} which have-
\begin{itemize}
\item only one fixed point, or
\item huge number of fixed points.
\end{itemize}
We now propose a (directed) graph, named {\em fixed point graph} (FPG), that facilitates the counting of fixed points of a given ACA.

\subsection{Fixed Point Graph}
The FPG of an ACA is a directed graph, where the vertices represent the passive RMTs of the ACA rule. To get FPG for an ACA, a forest considering the passive RMTs as individual vertices is formed. Now, we put a directed edge from vertex {\em u} to vertex {\em v}, if {\em u} and {\em v} are related following Table \ref{nextrmt}. For example, if RMTs 1, 2 and 5 are passive, then we can draw directed edges from vertex 1 to vertex 2, vertex 2 to vertex 5 and vertex 5 to vertex 2. But we can not draw a directed edge from vertex 1 to vertex 5, as RMT 1 and RMT 5 are not related (see Table \ref{nextrmt}).

\begin{figure}   
\centering     
\includegraphics[width=2.3in]{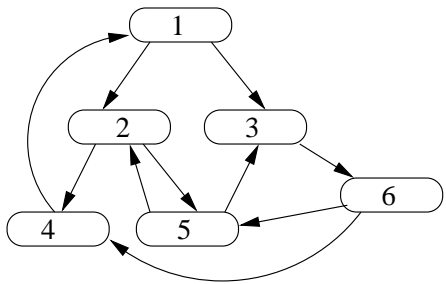}     
\caption{Fixed-point graph (FPG) of rule 77 ACA}  
\label{graphACA}  
\end{figure}   
         
\begin{eqed}
\label{exampleG}
This example illustrates the steps of constructing the FPG for rule 77 ACA. Here, the passive RMTs are 1, 2, 3, 4, 5 and 6. The vertices of the graph are 1, 2, 3, 4, 5 and 6 (see Figure \ref{graphACA}). Now considering the first vertex as RMT 1, we get the next RMTs of vertex 1 as RMT 2 and 3 (from Table \ref{nextrmt}). As these two RMTs are also vertices, we draw directed edges from vertex 1 to both the vertices 2 and 3 (see Figure \ref{graphACA}). Similarly, for vertex 3 the next possible RMTs are 6 and 7. But only one directed edge is possible from vertex 3 to vertex 6, since vertex 7 is absent (Figure \ref{graphACA}). After the  construction of directed edges for all the vertices, the graph is the desired FPG for rule 77 ACA (Figure \ref{graphACA}). 
\end{eqed}  

Using FPG, we can easily identify the fixed points of an ACA. To get a fixed point of an $n$-cell ACA, we start from a vertex of the FPG and check whether the vertex can be reached after visiting $n$ vertices (a vertex may be visited many times) including the starting vertex. If we can, then the sequence of vertices, visited, is the RMT sequence which represents a fixed point (Lemma \ref{lmf}).

\begin{eqed}
This example illustrates the steps for identifying fixed points for 4-cell rule 77 ACA.  Figure \ref{graphACA} is the FPG for rule 77 ACA. Starting from vertex 1, we can reach to vertex 1 again after visiting the intermediate vertices 3, 6 and 4 (see Figure \ref{graphACA}). So, the RMT sequence (1, 3, 4, 6) represents fixed point for the 4-cell rule 77 ACA, as total number of visited vertices is 4 (see Figure \ref{graphACA}). Similarly, starting from vertex 3 we can reach to vertex 3 again after visiting intermediate vertices 6, 4 and 1. Hence, the RMT sequence (3, 6, 4, 1) is another fixed point of the ACA. We can also get another two fixed points staring from vertex 6 (RMT sequence (6 ,4, 1, 3)) and vertex 4 (RMT sequence (4, 1, 3, 6)). We find two more fixed points starting from vertex 2 (RMT sequence 2, 5, 2, 5) and vertex 5 (RMT sequence (5, 2, 5, 2)). Hence, we identify total 6 fixed points of the 4-cell rule 77 ACA.
\end{eqed}

Therefore, fixed points of an ACA can be identified as well as counted utilizing the FPG. We can use the following recursive procedure for identifying/counting fixed points of an $n$-cell ACA--
\begin{itemize}
\item[] Start from each vertex, recursively visit next $n$ vertices. If the final vertex is the start vertex, consider that a fixed point is identified.
\end{itemize}
The above procedure correctly finds the number of fixed points of an $n-$cell ACA. However, it demands exponential time. Practically, for a moderate value of $n$, it is very difficult to count the number of fixed points.

\begin{figure*}[h]
\begin{center}
\includegraphics[width=4.8in, height=3in]{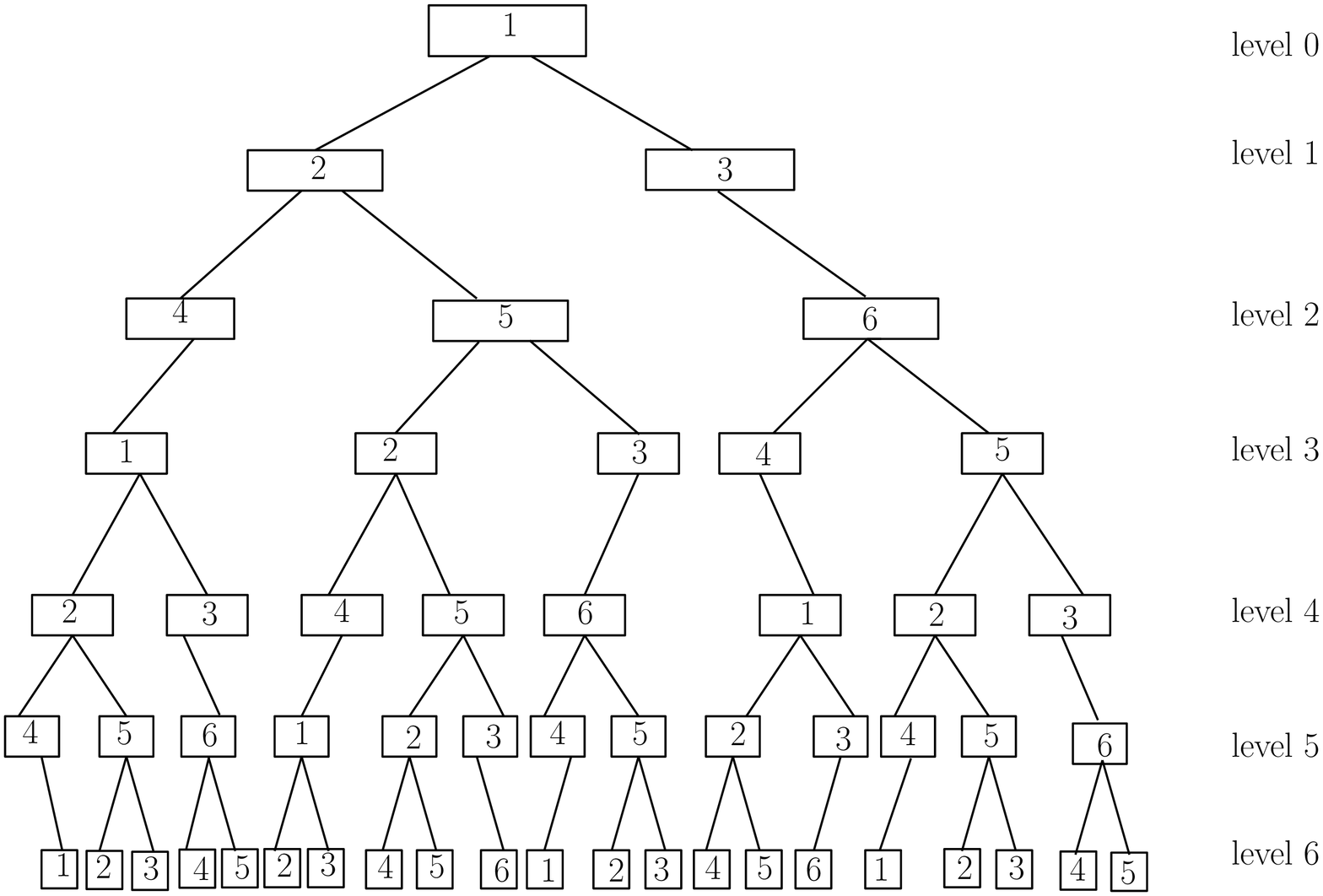}     
\caption{Recursion tree for 6-cell rule 77 ACA}
\label{tree1ACA} 
\end{center}  
\end{figure*}

\begin{figure*}
\begin{center}  
\includegraphics[width=4.8in, height=3in]{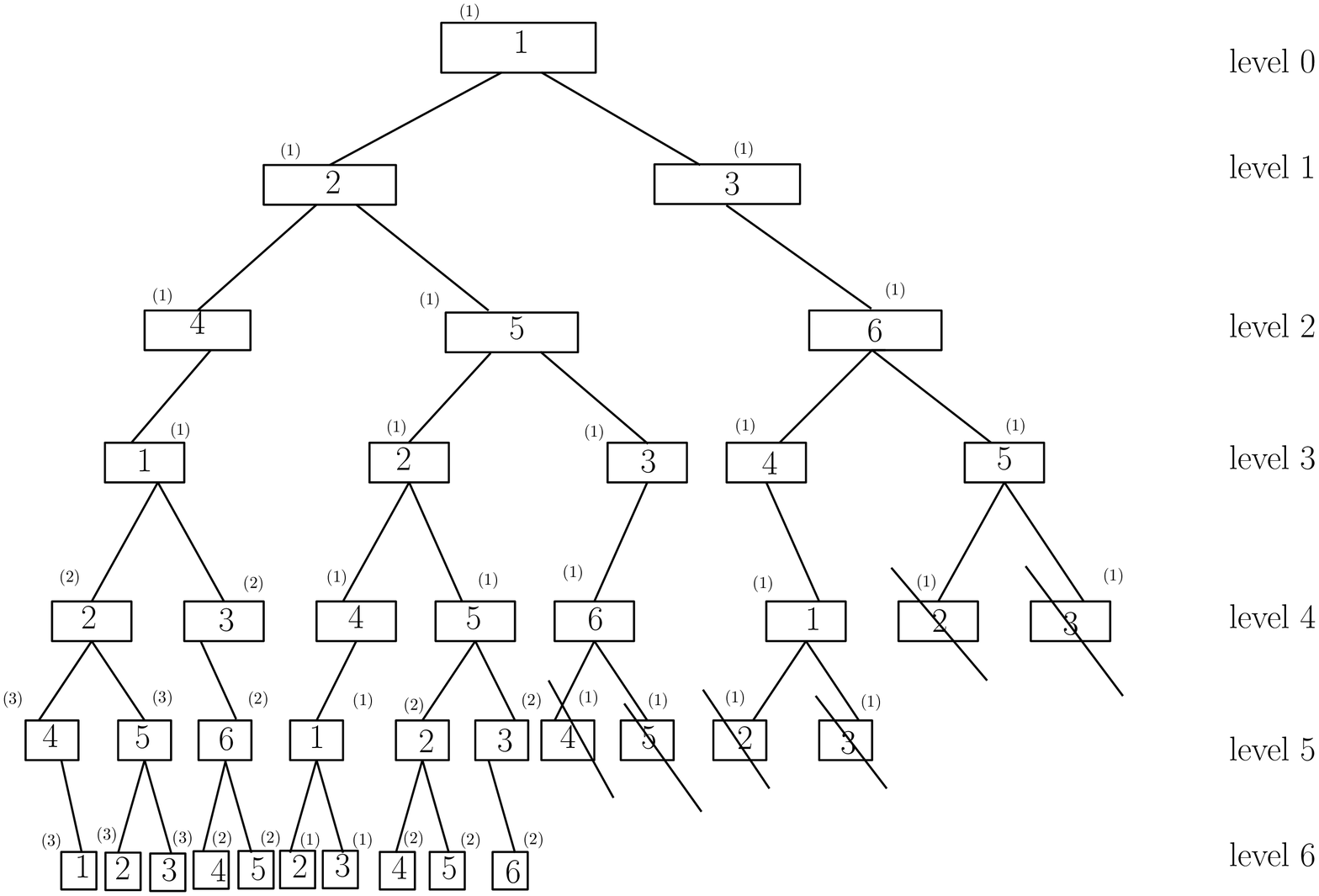}    
\caption{The tree for 6-cell rule 77 ACA after removing duplicate nodes}  
\label{treeACA} 
\end{center}  
\end{figure*} 

\subsection{The Counting Algorithm}

In order to understand the difficulty of the above scheme, we develop recursion tree of the procedure considering a node of the FPG as a starting vertex. Figure \ref{tree1ACA} shows the recursion tree of 6-cell rule 77 ACA. Obviously, the tree grows exponentially, which implies that the procedure demands exponential time to count the fixed points.

However, we observe that a node of the FPG can appear many times in a level of the tree. For example, in level 4 of figure \ref{tree1ACA}, node 2 and node 3 of figure \ref{graphACA} appear twice. However, the subtrees following the node 2 or node 3 are similar. In general, two subtrees, rooted at the same node of the FPG at any level, are similar. Moreover, two similar subtrees produce same number of fixed points. That is, if we get $x$ number of fixed points from one subtree, then another $x$ fixed points will be received from another similar subtree. For example, each of the subtrees following the node 2 at level 4 of figure \ref{tree1ACA} produces one fixed point. So, we need not to proceed with all duplicate nodes, but with a single one, after assigning appropriate {\em weight} to the node of the tree. This {\em weight} of nodes of the tree plays a crucial role to drastically reduce time complexity of the proposed counting algorithm.

However, the number of fixed points present in an $n$-cell ACA can be decided after reaching at leaves of the tree. If a node at level $n$ is same with the root, we consider a fixed point is received. In this way, we get three fixed points from figure \ref{tree1ACA} (for 6-cell rule 77 ACA). 

The sequence of nodes from the root to a leaf of the recursion tree can represent at most one fixed point. So, we initially consider that the weights of nodes of the tree are one. When we find a single node of the FPG appears twice or more in a level of the tree, we keep one node in the tree having its weight as sum of weights all of its duplicates and its own. Figure \ref{treeACA} shows the tree after assigning weights to the nodes of figure \ref{tree1ACA}. The weights are noted within brackets, and the duplicates are dropped in the figure. In figure \ref{treeACA}, there is only one leaf, which is same with the root (according to the FPG). The weight of the leaf is 3. Hence, we conclude that the tree can shown three fixed points.

In the proposed algorithm, we develop a tree, similar to recursion tree, considering each node of the FPG as the root. However, the algorithm does not store the whole tree. Rather, it stores the nodes of a single level. Initially, we assign the weight of the root as one. When a child is generated from a parent, the child copies the weight of its parent. As we have mentioned before, if we find duplicate nodes in a level of the tree, we proceed with one having weight equals to the summation of weights of its duplicates (including itself).

The algorithm (Algorithm \ref{algocount}) takes FPG of an ACA, and the number of cells ($n$) of the ACA as input, and outputs the number of fixed points, present in the $n$-cell ACA. The algorithm uses a variable, $FP$ (which is initialized to 0) to count the fixed points. However, for each vertex of the FPG,  we develop a tree (Step 1). Step 6 find the nodes of next level of the tree, whereas Step 8 removes the duplicates from the generated nodes. We use a variable, $m$ to note the number of unique nodes in a level. The variable $w_j$ stores the weight of $j^{th}$ ($1\le j\le m$) unique node (from the set of $m$ unique nodes). Finally, the algorithm reports the number of possible fixed points of the ACA.

\begin{algorithm}
 \caption{Count Fixed Points} 
 \begin{algorithmic}[1]
 \label{algocount}
 \renewcommand{\algorithmicrequire}{\textbf{Input:}}
 \renewcommand{\algorithmicensure}{\textbf{Output:}}
 \REQUIRE FPG of an ACA, ACA cell size ($n$)
 \ENSURE   Number of Fixed points
 \\ \textit{Initialization} : $FP$ = 0
  \STATE  For each vertex $u$ of the FPG, repeat Step 2 to Step 12 
  \STATE  $m$ = 1, $S_1$ = $u$, $w_1$ = 1 
  \FOR {$i = 1$ to $n-2$}
  \STATE $k$ = 0
   \FOR {$j =1$ to $m$}    
   \STATE For each $v$ when ($S_j$, $v$) exist, set $k$=$k$+1, $S_k'$ = $v$ and $w_k'$ = $w_j$
   \ENDFOR
   \STATE For any $p$ ($\leq$ $k$) and $q$ ($\leq$ $k$) when $p$ $\neq$ q and\\$S_p'$ = $S_q'$, remove $S_q'$ and set $w_p'$ = $w_p'$ +  $w_q'$
  \STATE Set $m$ as the number of vertices in $S'$
  \STATE Assign $S_k$ = $S_k'$ for all $k$= 1, 2, $\cdots$ $m$
  \ENDFOR
  \STATE  For each $j$ (1 $\leq$ $j$ $\leq$ $m$) if ($S_j$, $u$) exist, then $FP$ = $FP$ + $w_j$  
  \STATE  Print $FP$ as number of fixed points. 
 \end{algorithmic}
 \end{algorithm}
 
\begin{eqed}
This example illustrates the counting of fixed points of 6-cell rule 77 ACA using Algorithm \ref{algocount}. Considering node 1 of the FPG (Figure \ref{graphACA}) as the root, we get the tree of Figure \ref{treeACA}. For this tree, we get 3 fixed points. When node 2 is the root, we find another 6 fixed points. In this way, we can finally get that there are 20 fixed points in the 6-cell rule 77 ACA. 
\end{eqed} 

\subsection*{Complexity:} The complexity of Algorithm \ref{algocount} depends on a nested loop, which includes three {\bf for} loops at lines 1, 3 and 5. The {\bf for} loop at line 5 depends on ${\bf m}$, which is dependent on the number of vertices of the FPG. Since, the number of vertices of a given FPG is fixed (it can be at most 8), the time requirement of each execution of this loop is $O(1)$. However, the loop of line 3 is clearly dependent on $n$, which implies that, the time requirement of the steps 3 to 7 is  $O(n)$. Steps 8 to 10 require $O(1)$ time, and the most outer loop (at line 1) depends on the number of vertices of an FPG, which is a constant for a given FPG. Hence, complexity of Algorithm \ref{algocount} is $O(n)$.

Out of 256, 146 ACAs in two-state three-neighborhood interconnection that always approach towards fixed points are already listed in Table \ref{AFP} (Theorem \ref{fixed}). However, we can further identify a set of ACAs which are having multiple attractors. Algorithm \ref{algocount} guides us to identify such set of ACAs. There are 84 ACAs which are having multiple fixed points. We further eliminate ACAs 76, 140, 196, 200, 220, 205, 206 and 236 from the list of 84 ACAs. These ACAs have only one RMT active and rest 7 RMTs are passive. So, these ACAs have huge number of fixed points for a given $n$. These ACAs are not suitable candidates as pattern classifier. Rule 204 ACA is the special rule where all RMTs are passive. All the states of rule 204 ACA are fixed points. Hence, we eliminate rule 204 ACA as a pattern classifier. Table \ref{MACA} shows the list of ACAs with multiple fixed points after eliminating the above noted rules. 

\begin{table}   
\caption{ACAs with multiple fixed points}  
\label{MACA}    
\centering
\begin{tabular}{|cccccccccc|}\hline   
4&5&12&13&36&44&68&69&72&77 \\
78&79&92&93&94&95&100&104&128&130 \\
132&133&136&138&141&144&146&152&154&160 \\
162&164&166&168&170&172&174&176&178&180 \\
182&184&186&188&190&192&194&197&202&203 \\
207&208&210&216&217&218&219&221&222&223 \\
224&226&228&232&234&237&238&240&242&244 \\
246&248&250&252&254&&&&& \\  \hline 
\end{tabular}  
\end{table}  

However, we further identify and eliminate ACAs from Table \ref{MACA}, whose convergence time is exponential. These ACAs are not better candidate as pattern classifier. In the next section, we identify such ACAs and eliminate them.

\section{ACAs with exponential convergence time}  
\label{ele}

In this section, we report the method of finding the average convergence time of ACA.
In our earlier work \cite{automata2013}, we have experimentally studied the convergence time of ACAs. We have simulated ACA rules to find their average convergence time and studied the rate of growth of convergence time with respect to the size of automaton. However, we have not studied,  the average convergence time of ACA with different {\em update patterns} from a particular initial configuration. In this work, we consider different {\em update patterns} for finding of average convergence time of ACA from a particular initial configuration. 

\subsection{Experimental setup}
\label{ES}
The convergence time of an ACA depends on both, initial configuration and {\em update pattern}. During calculation of convergence time, however, we consider $n$ updates as a single time step.

In the reported experiment we use simple random sampling {\em with replacement} to estimate the population mean ($\mu$). We take a series of samples of size $m$ in the estimation process. Consider, ${\overline{X_k}}$ denotes the mean of $k^{th}$ sample, and ${\widehat{\overline{X_k}}}$ denotes the $k^{th}$ estimate to the population mean ($k\ge 1$). Obviously, ${\overline{X_k}}=\frac{1}{m}\sum_{i=1}^{m}{x_i}$, where $x_i$ is an element of the population chosen randomly and uniformly. However, the ${\widehat{\overline{X_k}}}$ is determined in the following way.
\begin{multline}
\label{method}
\begin{array}{l}
\widehat{\overline{X_1}}={\overline{X_1}} \\   
\widehat{\overline{X_2}}=\frac{{\overline{X_1}}+{\overline{X_2}}}{2}=\frac{1}{2}\widehat{\overline{X_1}}+\frac{1}{2}{\overline{X_2}}\\
\widehat{\overline{X_3}}=\frac{{\overline{X_1}}+{\overline{X_2}}+{\overline{X_3}}}{3}=\frac{2}{3}\frac{({\overline{X_1}}+{\overline{X_2}})}{2}+\frac{1}{3}{\overline{X_3}}=\frac{2}{3} \widehat{\overline{X_2}}+\frac{1}{3}{\overline{X_3}}\\
\dots\\
\widehat{\overline{X_k}}=\frac{k-1}{k}\widehat{\overline{X_{k-1}}}+\frac{1}{k}{\overline{X_1}}
\end{array}
\end{multline}
  
Since the mean of all possible samples' means is the population mean, the series $({\widehat{\overline{X_k}}})_{k\in {\mathbb{N}}}$ approaches to $\mu$. Population size in our study is generally large. So, neither consideration of all possible samples nor finding of $\mu$ is possible. We, therefore, declare ${\widehat{\overline{X_k}}}$ as our final estimate to the population mean if $\frac {|\widehat{\overline{X_k}}-\widehat{\overline{X_{k-1}}}|}{\widehat{\overline{X_k}}} < \delta$, where $\delta$ is a small threshold value. The $\delta$ specifies the precision we desire to achieve. We consider $\delta = 0.01$.

However, in our experiment estimation of sample size $m$ is an issue. Consider, we wish to control the {\em relative error} $r$ in ${\overline{X_k}}$, such that
\begin{equation*}
\Pr\left ( \left |\frac{{\overline{X_k}}-\mu}{\mu}\right |\ge r \right ) = \alpha
\end{equation*}
where $\alpha$ is a small probability. If we now assume that ${\overline{X_k}}$ is normally distributed, we can get the estimated value of $m$ \cite{William}.
\begin{equation}
\label{SampleSize}
  m = \frac{{t^2}{S^2}}{{r^2}{\mu^2}}       
\end{equation} 
where $S^2$ and $\mu$ are the population variance and mean respectively, and $t$ is a parameter related to $\alpha$. In our experimentation, we consider $\alpha=0.05$ which yields to $t=2$, and $r=0.1$ \cite{William}. However, $S^2$ and $\mu$ both are unknown in our case. So we need to again estimate these two parameters. The most reliable method to do this is, take a sample of size $m_1$ and estimate $S^2$ and $\mu$ for the use of Equation~\ref{SampleSize} \cite{William}. We have taken $m_1=100$ in this study. It may be noted that in all cases sample size is much less than population size.

Above method is used to estimate average convergence time of an ACA. It is worthwhile to mention that though we allow error with a small probability $\alpha$, the error in estimate, if any, is reduced due to the use of the method noted in Relation~\ref{method}. For each convergence time, we randomly and uniformly choose a configuration (with replacement) from $2^n$ possible configurations of an $n$-cell ACA, and then find the time using fully asynchronous updating scheme. Before estimating the average convergence time of an ACA having $n$ cells, however, we first estimate the sample size $m$. Following example illustrates the estimation of $m$ for an ACA.
\begin{eqed}  
Consider $m_1=100$, $t=2$ and $r=0.1$. To estimate sample size $m$ for ACA 226 with $n=20$, we first choose 100 configurations randomly and uniformly (with replacement) from $2^{20}$ possible configurations. Then, we find convergence time of ACA 226, updated under fully asynchronous mode, against each of the 100 configurations. We observe in an experiment that convergence time vary from 1 to 192 with estimated $S^2=1762$ and estimated $\mu = 41$. Hence, using Equation~\ref{SampleSize} we get $m=\frac {2^2\times1762}{(0.1)^2\times (41)^2}\approx 419$.
\end{eqed}

It is, however, already mentioned that convergence time of an ACA depends not only on initial configuration but also on {\em update pattern}. That is, for each initial configuration, we may get different convergence time for different {\em update patterns}. Since there are a huge number of possible ways for an ACA to converge from an arbitrary initial configuration, we again use the above method to estimate average convergence time for each initial configuration. Whole scheme that is used to estimate average convergence time of an $n$-cell ACA is noted in the following Algorithm \ref{algo}.
\begin{algorithm}
 \caption{Find Convergence Time }
 \begin{algorithmic}[1]
 \label{algo}
 \renewcommand{\algorithmicrequire}{\textbf{Input:}}
 \renewcommand{\algorithmicensure}{\textbf{Output:}}
 \REQUIRE ACA rule, $n$ (size of ACA)
 \ENSURE  Average convergence time 
 \\ \textit{Initialization} : $m_1 = 100$, $\delta = 0.01$
  \STATE Choose a configuration ($\mathcal C$) randomly and uniformly from $2^n$ possible configurations.
   \STATE Find convergence time of the ACA against $\mathcal C$ updated under fully asynchronous mode.
  \STATE Repeat {\em Steps} 1 and 2 for $m_1$ times, and estimate variance ($S^2$) and mean ($\mu$) to convergence time. 
    \STATE Estimate sample size $m$ using Equation~\ref{SampleSize}. Consider $t=2$ and $r=0.1$.
    \STATE Choose a configuration ($\mathcal C$) randomly and uniformly from $2^n$ possible configurations.
    \STATE Find convergence time for $m_1$ times with the same initial configuration ($\mathcal C$). That is, we are considering $m_1$ update patterns for the $\mathcal C$. Estimate sample size $m'$ like {\em Step 4}.
    \STATE Find convergence time for $m'$ times with $\mathcal C$. Get the mean value of this sample.
    \STATE Get another mean using {\em Step 7}. If the difference between two consecutive means is less than $\delta\times \mbox{last mean}$, we consider the last mean as the average convergence time of the ACA for the $\mathcal C$. Otherwise, repeat {\em Step 8}.
    \STATE Repeat {\em Steps} 5 to 8 for $m$ times to get $m$ convergence time. Calculate the average of these convergence time.
   \STATE Repeat {\em Step 9} to get another average convergence time. If the difference between two consecutive average convergence time is less than $\delta\times \mbox{last average time}$, we consider the last average time as the estimated convergence time. Otherwise, repeat {\em Step 10}.\\
\STATE  Output the estimated convergence time.                
\end{algorithmic}
\end{algorithm}

\subsection{The results} 
\label{MCT}
We now use Algorithm~\ref{algo} to find average convergence time of all 
146 ACAs of Table \ref{AFP}. Since we are interested to establish a relation between convergence time and $n$, the ACA size, we find a series of average convergence time of an ACA for different values of $n$. As $n$ increases, average convergence time of an ACA (except ACA 204) increases. To find the rate of growth of convergence time, we use the {\em empirical curve bounding technique} \cite{Catherine}. However, we estimate here the upper bound of convergence time. We declare that the convergence time $T(n)$ has upper bound $g(n)$ (we write $T(n)=O(g(n))$) if there exist two positive constants $c$ and $n_0$ such that $0\leq T(n)\leq c~g(n)$, for all $n\geq n_0$ \cite{Catherine}. 

To estimate the upper bound, we assume that the convergence time ($t$) follows power rule, that is, $t \approx kn^a$ \cite{Catherine}. We next approximate the coefficient $a$ by taking empirical measurements of time $t_1$ and $t_2$ for ACA size $n_1$ and $n_2$ respectively. Hence, we can get $\frac {t_2}{t_1}\approx (\frac {n_2} {n_1})^a$, and
\begin{equation}
\label{rate}
  a\approx \frac{\log (t_2/t_1)}{\log(n_2/n_1)}
\end{equation}
Now, our task is to experimentally find a series of $t$ values for different $n$, and then using Equation~\ref{rate} we estimate the $g(n)$.

Tables~\ref{logn} -- \ref{exponential} show the experimental results of 13 representative ACAs. The tables show the convergence time ($t$) of ACAs for different $n$, and values of $a$ (Equation~\ref{rate}). Note that the ACAs of each table have similar convergence time.

For Table~\ref{n2}, the value of $a$ for each ACA is nearly 2 after some value of $n$. So we estimate $g(n)=n^2$. Hence we get the average convergence time of these ACAs as $O(n^2)$. This is validated in figure~\ref{figrate}(a), where the red curve denotes $0.2~n^2$ (that is, $c=0.2$) and other 6 curves are for 6 ACAs -- 138, 146, 170, 178, 194 and 226.

In Table~\ref{logn}, the $a$ values are ever decreasing and they are generally less than 1. The trend of $a$ values suggests that for these ACAs of Table~\ref{logn}, $g(n)=\log n$. We always consider that the base of $\log$ is 2. So average convergence time of these ACAs is $O(\log n)$. We validate this estimation in figure~\ref{figrate}(c), where the red curve shows the upper limit which is $3.5\log n$ (that is, $c=3.5$), and other 6 curves are for 6 ACAs -- 130, 192, 202, 206, 234 and 242. 

Table~\ref{exponential} contains only ACA 210, whose rate of growth of convergence time is high. Our guess is, convergence time for this case is exponential. To validate this guess, we take log of convergence time, and replace the $t$ values of Equation~\ref{rate} by log of $t$ values. So we get a new $a$ values, say they are $a'$. Table~\ref{exponential} shows both the values, $\log t$ and $a'$. After some values of $n$, $a'$ is nearly 1. So the value of $\log t$, almost linearly increases with $n$. Hence, we estimate $g(n)=2^n$. That is, average convergence time is $O(2^n)$. This is validated in figure~\ref{figrate}(b).

\begin{figure*}
\centering  
\begin{tabular}{@{\hspace{-0.1in}}c@{\hspace{-0.2in}}c@{\hspace{-0.2in}}c}     
\includegraphics[scale=0.45]{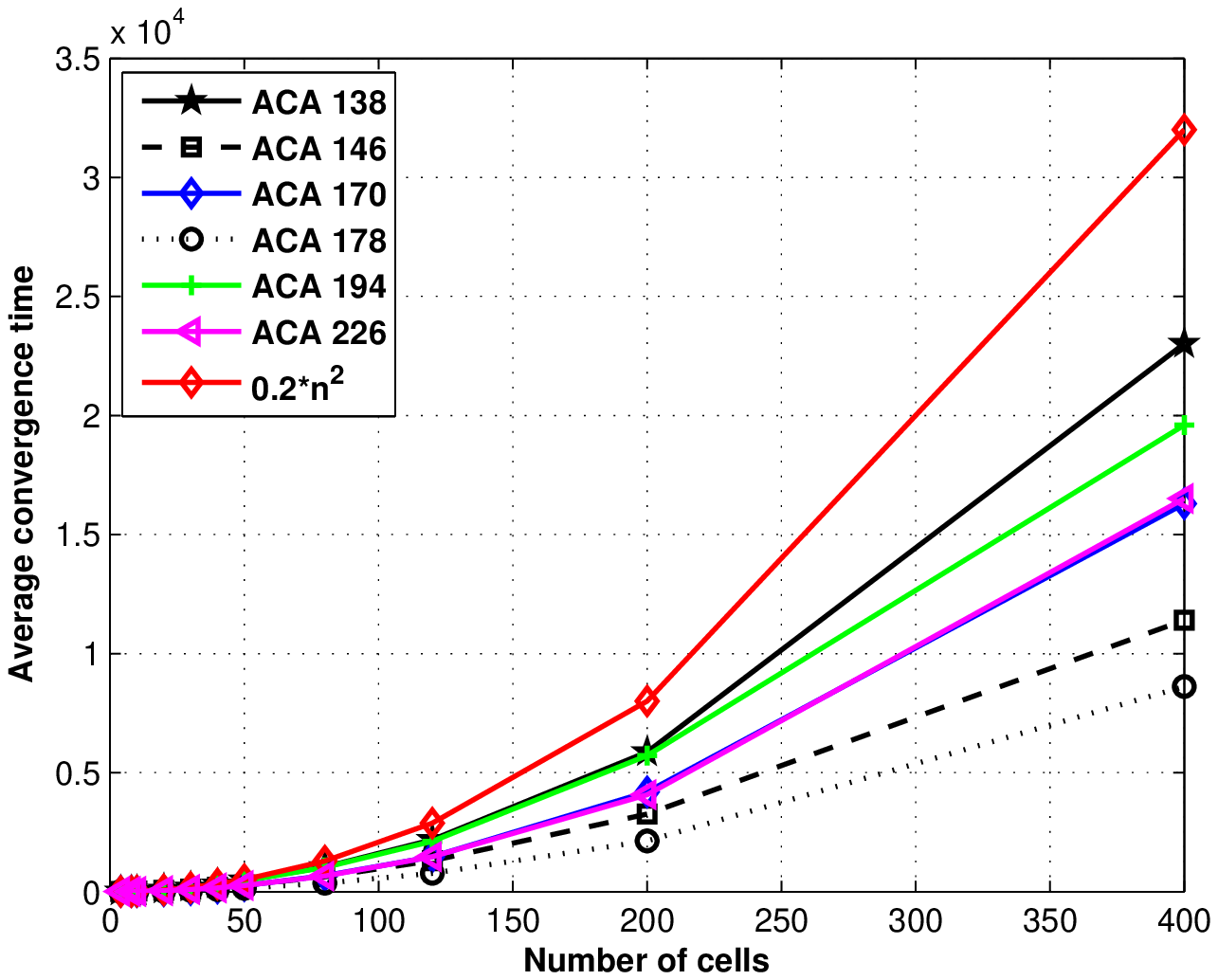} & \includegraphics[scale=0.45]{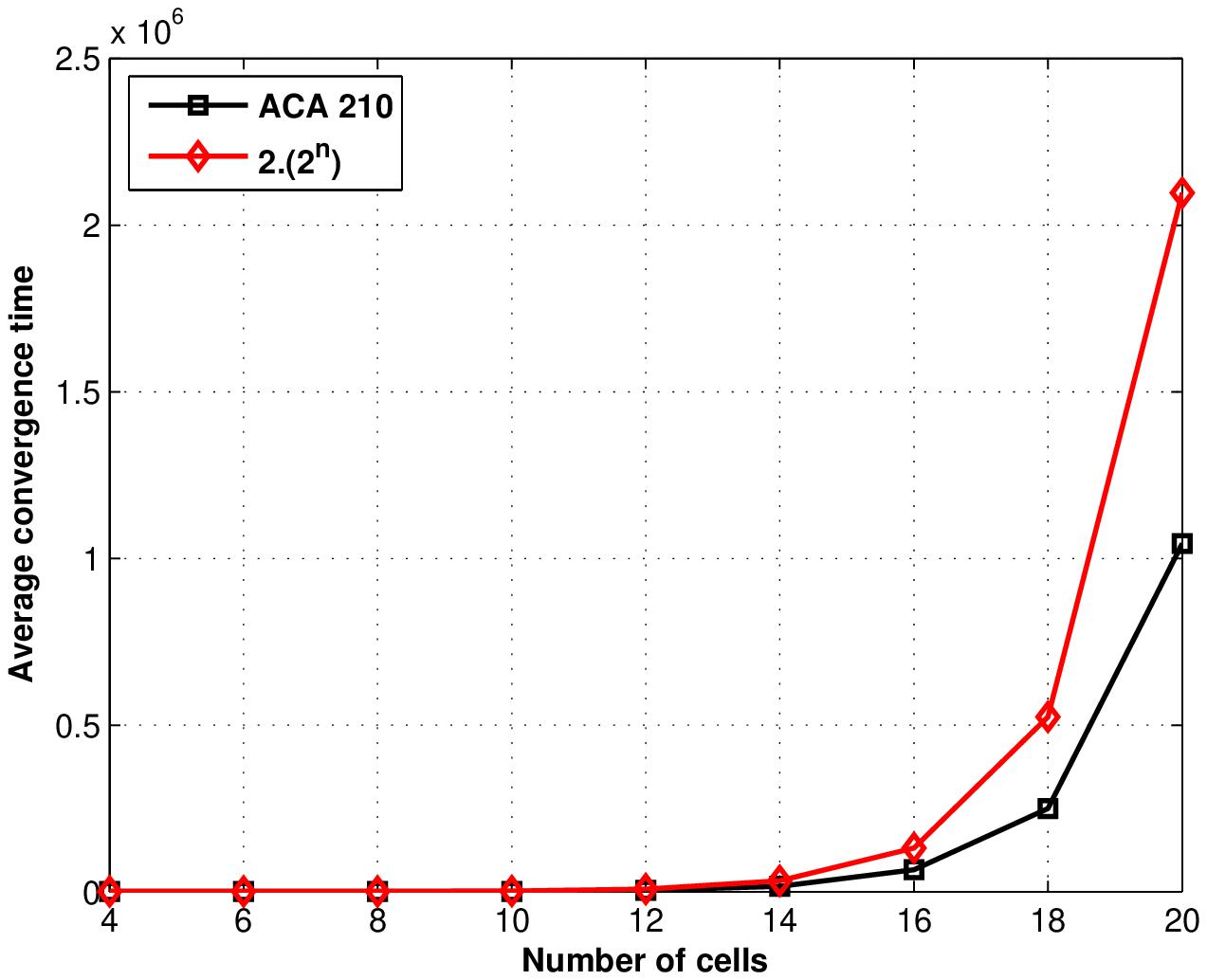} \\
(a)&(b)\\
\multicolumn{2}{c}{\includegraphics[scale=0.43]{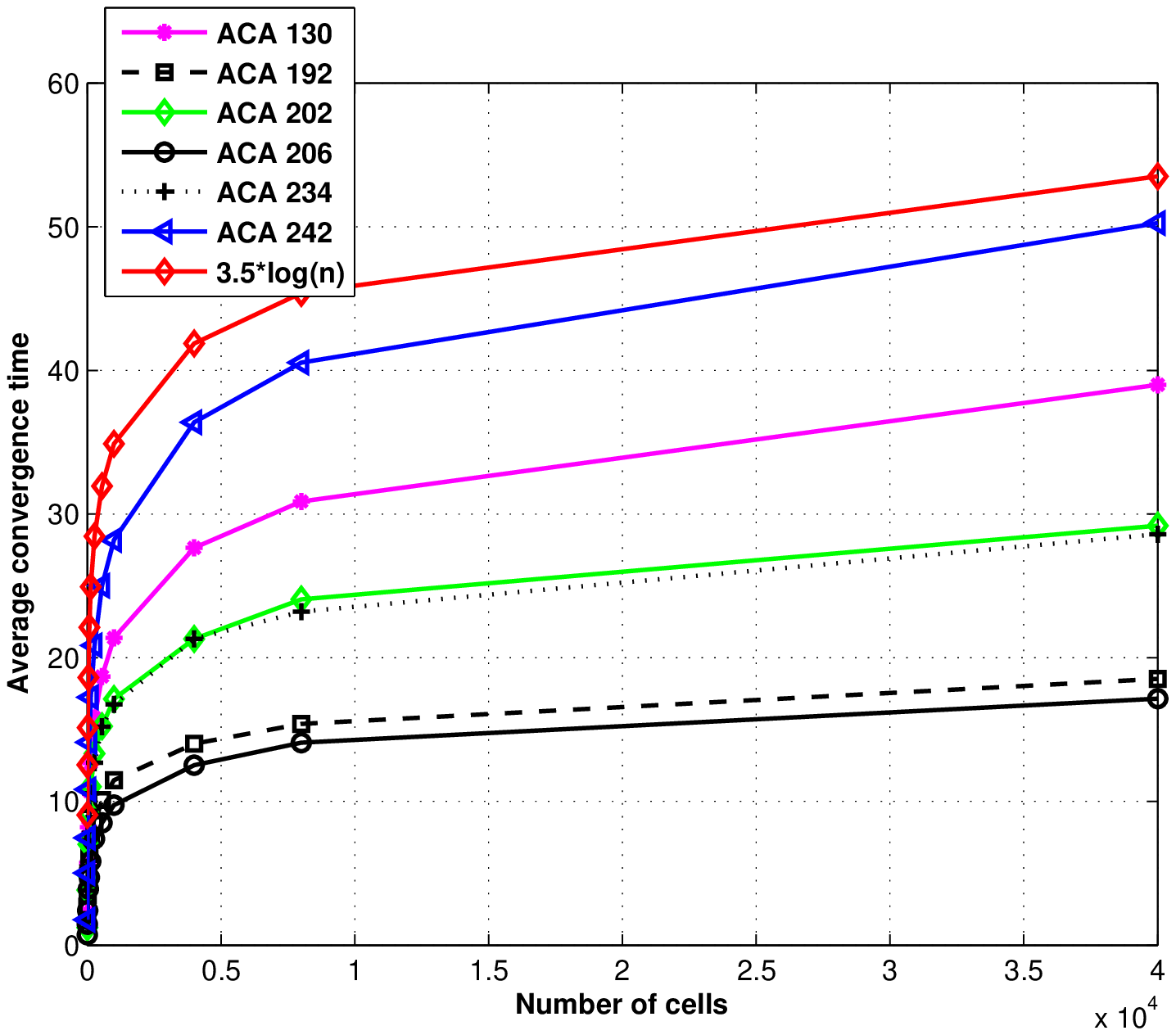}} \\
\multicolumn{2}{c}{(c)} 
\end{tabular} 
\caption{Convergence time of ACAs having upper bound  (a) $n^2$ (b) $2^n$ (c) $\log n$} 
\label{figrate}  
\end{figure*} 
The average convergence time for rule 204 ACA is constant as all the RMTs for rule 204 ACA are passive. Hence, the convergence time for rule 204 ACA is $O(1)$. The rest 145 ACA rules are categorized in 4 classes based on their upper bound of convergence time. All 146 ACA are listed with their rate of growth of convergence time in the first column of Table \ref{ctime}. Therefore, there are mix ACAs--90, 122, 154, 161, 165, 166, 180 and 210, which are having exponential average convergence time. We consider none of these ACAs as potential classifier.

\begin{table*}[h]
\caption {ACAs with convergence time $O(\log n)$}
\centering
\label{logn}
\begin{tabularx}{0.98\textwidth}{|c *{7}{|Y}|}
\cline{1-7}
 \multicolumn{1}{|c|}{\# cells} 
 & \multicolumn{2}{c|}{ACA 130}  
 & \multicolumn{2}{c|}{ACA 192}
  & \multicolumn{2}{c|}{ACA 202} \\ \cline{2-7}       

 n & $t$ & $a$ & $t$ & $a$ & $t$ & $a$ \\ \hline
6 &2.27&  &1.74& &1.22 &     \\  \hline 
12 & 4.02 & 0.82 & 2.66 & 0.83 &3.83&1.65 \\ \hline
20&5.72&0.69&3.49&0.39& 4.89&0.47 \\ \hline
40&8.20&0.51&4.99&0.51&7.01&0.51  \\ \hline
80&10.71&0.38&6.35&0.34&9.12&0.37 \\ \hline
140&12.66&0.29&7.36&0.26&11.02&0.33 \\ \hline
280&15.80&0.31&8.44&0.19&13.34&0.27  \\ \hline
560&18.70&0.24&10.11&0.26&15.24&0.19 \\ \hline
1000&21.37&0.23&11.49&0.22&17.12&0.20 \\ \hline
4000&27.65&0.18&14.01&0.19&21.32&0.15  \\ \hline
8000&30.87&0.15&15.39&0.13&24.07&0.17  \\ \hline
40000&39.00&0.14&18.53&0.12&29.20&0.12  \\ \hline
 \multicolumn{1}{|c|}{\# cells} 
 & \multicolumn{2}{c|}{ACA 206}  
 & \multicolumn{2}{c|}{ACA 234}
  & \multicolumn{2}{c|}{ACA 242} \\ \cline{2-7}      
n & $t$ & $a$ & $t$ & $a$ & $t$ & $a$ \\ \hline
6 & 0.73 &  &1.65&  &1.76 &   \\ \hline
12 &1.46& 1 &3.78&1.19&5.03&1.51 \\ \hline 
20&2.4&0.97&5.25&0.64&7.46&0.77  \\ \hline
40 &3.91&0.70&7.29&0.47&10.85&0.54 \\ \hline
80 &4.72&0.27&9.34&0.35&14.11&0.37 \\ \hline 
140&5.84&0.38&10.58&0.22&17.24&0.35 \\ \hline
280&7.39&0.33&12.68&0.26&20.87&0.27 \\ \hline
560 &8.52&0.20& 1521 &0.26&25.01&0.26 \\ \hline
1000&9.74&0.23&16.74&0.16&28.20&0.20 \\ \hline
4000&12.53&0.18&21.32&0.17&36.38&0.18 \\ \hline
8000 &14.08 & 0.16 & 23.21 & 0.12&40.54 & 0.15 \\ \hline
40000 &17.16 & 0.12 & 28.59 &0.12 &50.25 &0.13 \\ \hline
\end{tabularx}
\end{table*}  
 
\begin{table*}[h]
\caption {ACAs with convergence time $O(n^2)$}
\centering
\label{n2}
\begin{tabularx}{0.98\textwidth}{|c *{7}{|Y}|}
\cline{1-7}
 \multicolumn{1}{|c|}{\# cells} 
 & \multicolumn{2}{c|}{ACA 138}  
 & \multicolumn{2}{c|}{ACA 146}
  & \multicolumn{2}{c|}{ACA 170} \\ \cline{2-7}       

 n & $t$ & $a$ & $t$ & $a$ & $t$ & $a$ \\ \hline
4&3.080&  & 2.39& & 0.53 &  \\ \hline
8&16.010&2.41&11.14&2.22&5.47&3.36  \\ \hline
10&24.049&1.82&16.93&1.87&8.90&2.18 \\ \hline
20&89.54.&1.89&54.57&1.68&39.02&2.13 \\ \hline
30&177.669&1.68&109.01&1.70&89.61&2.05 \\ \hline
40&229.929&1.82 &185.17&1.84&162.11&2.06 \\ \hline
50&466.399&1.97&279.0&1.83&249.88&1.93 \\ \hline
80&1032.280&1.69&618.94&1.69&671.96&2.10 \\ \hline
120&2177.929&1.84&1298.06&1.83& 1476&1.94  \\ \hline
200&5867.109&1.93&3255.46&1.79&4182&2.03 \\ \hline
400&23018.94&1.97&11401.88&1.80&16294&1.96 \\ \hline
\multicolumn{1}{|c|}{\# cells} 
 & \multicolumn{2}{c|}{ACA 178}  
 & \multicolumn{2}{c|}{ACA 194}
  & \multicolumn{2}{c|}{ACA 226} \\ \cline{2-7}      
n & $t$ & $a$ & $t$ & $a$ & $t$ & $a$ \\ \hline
4&0.53& &2.84& &0.53& \\ \hline
8& 4.13&2.96&15.4&2.43&5.44&3.35\\ \hline
10&6.74&2.19&22.15&1.62&8.94&2.22\\ \hline
20&24.78&1.87&82.15&1.89&38.90&2.12\\ \hline
30&54.57&1.94&173.96&1.85&88.12&2.00\\ \hline
40&94.69&1.91&283.64&1.69&161.02&2\\ \hline
50&146.50&1.995&459.60&2.16&263.66&2.20\\ \hline
80&363.32&1.93&1014.91&1.68&658.77&1.94\\ \hline
120&777.40&1.87&2121.05&1.81&1470&1.97\\ \hline
200&2129.06&1.97&5720.00&1.94&4064&1.99\\ \hline
400&8615.53&2.01&19601.85&1.77&16514.31&2.02\\ \hline
\end{tabularx}
\end{table*}  
 
\begin{table}[h]
\caption{ACA with convergence time $O(2^n)$}
\centering    
\label{exponential} 
\scalebox{1.0}{ 
\begin{tabular}{|c|c|c|c|c|}\hline   
\multicolumn{5}{|c|}{ACA 210} \\ \hline
\# cell & $t$& $a$ & $\log t$& $a'$ \\  
n & & & & \\ \hline
4 &6.75&  & 2.75&  \\ \hline
6&44.66&4.66&5.48&0.58 \\ \hline
8&210.72&5.39&7.71&0.84 \\ \hline
10&889.51&6.45&9.79&0.93 \\ \hline
12&3629.33&7.71&11.82&0.96 \\ \hline
14&15409.08&9.37&13.91&0.94 \\ \hline
16&65297.30&10.81&15.99&0.95\\ \hline
18&249072.09&11.36&17.92&1 \\ \hline
20&1044845.12&13.60&19.99&0.96 \\ \hline 
\end{tabular} }    
\end{table}  

\begin{table}[h]
\caption{Convergence time of ACAs}  
\label{ctime}    
\centering 
\begin{tabular}{|c|c|} \hline   
Rate of & ACAs\\ 
growth& \\ \hline 
$O(1)$&204 \\ \hline
 $O(log~n)$   & 0, 2, 4, 5, 8, 10, 12, 13,  16, 18, 24,\\
              & 32, 34, 36, 40, 42, 44, 48, 50, 56, 64, 66,\\
              & 68, 69, 72, 76, 78, 77, 79, 80, 92, 93, 94, \\
              & 95, 96, 98, 100, 104, 112, 128, 132, 133 ,130,\\
              & 136, 140, 141, 144, 160, 162, 164, 168, 171, 172,\\
              & 175, 176, 179, 183, 185, 186, 187, 189, 190, 191,\\
              & 192, 196, 197, 200, 202, 203, 205, 206, 207, 216, \\
              & 217, 218, 219, 220, 221, 222, 223, 224, 227, 228, \\
              & 231, 232, 233, 234, 235, 236, 237, 238, 239, 241, \\
              & 242, 243, 245, 246, 247, 248, 249, 250, 251, 252,  \\
              & 253, 254, 255  \\  \hline
$O(n^{1/2})$  & 26, 58, 74, 82, 88, 106, 114, 120, 163, 167, \\
              & 169, 173, 177, 181, 225, 229 \\ \hline
$O(n^2)$ & 138, 146, 152, 170, 174, 178, 182, 184, 188, 194, \\
         & 208, 226, 240, 244, 230  \\ \hline
$O(2^n)$ & 90, 122, 154, 161, 165, 166, 180, 210 \\ \hline
\end{tabular}  
\end{table} 

\begin{table}[H]   
\caption{Candidate ACAs as pattern classifier}  
\label{FLP}    
\centering
\begin{tabular}{|cccccccccc|}\hline   
4&5&12&13&36&44&68&69&72&77 \\
78&79&92&93&94&95&100&104&128&130 \\
132&133&136&138&141&144&146&152&160&162 \\
164&168&170&172&174&176&178&182&184&186 \\
188&190&192&194&197&202&203&207&208&216 \\
217&218&219&221&222&223&224&226&228&232 \\
234&237&238&240&242&244&246&248&250&252  \\
254&&&&&&&&& \\  \hline 
\end{tabular}  
\end{table}

\section{Performance of Pattern Classifier}    
\label{per} 
The design of ACA based pattern classifier is already reported in Section \ref{design}. It has been pointed out that an ACA can act as pattern classifier if - (1) the ACA converges to a fixed point from an arbitrary seed, (2) the convergent ACA is having at least two but not huge number of fixed points, and (3) the convergence time is not exponential. Table \ref{FLP} shows the list of ACAs (71 ACAs) which satisfy the above demands. However, this section presents the performance analysis of these ACAs as pattern classifiers. We have considered standard and widely accepted data sets available at {\em http://www.ics.uci.edu/$\sim$mlearn/MLRepository.html} for the study of efficiency of the classifier. To understand the effectiveness of the classifier, it must be passed through two phases--{\em training phase} and {\em testing phase}.  
 
\subsection{Training Phase} 
All the 71 candidate ACAs can act as pattern classifier. However, to find the best effective classifier, we need to train all the candidate ACAs in this {\em training phase}. All these ACAs are trained with patterns of two different data sets. We use Algorithm \ref{algocla} for the training of all these candidate ACAs. The ACA with highest efficiency is considered as the desired classifier. It takes two data sets ({\em$P_1$} and {\em $P_2$}) and a list of candidate ACAs (Table \ref{FLP}) as its input. In this training phase, an ACA of Table \ref{FLP} is loaded with patterns of two different data sets ({\em$P_1$} and {\em $P_2$}) and updated till the ACA converge to a fixed point attractor. We store all the attractors and count the number of patterns converged to them. An attractor is declared as of class 1 if more patterns from pattern set {\em$P_1$} than that of {\em$P_2$} are converged to the attractor and the attractor is stored in {\em attrset-1}, otherwise the attractor is of class 2 and stored in {\em attrset-2}. The efficiency of an ACA is determined using the following: 
 
\begin{equation}
\label{effi}
Efficiency=\frac{\sum\limits_{i=1}^m \max(n_{1}^i, n_{2}^i)}{|P_1| + |P_2|}.
\end{equation}

Where, $n_{1}^i$ and $n_{2}^i$ are the maximum number of patterns converged to the $i^{th}$ fixed point attractor of an ACA from data set $P_1$ and $P_2$ respectively. $|P_1|$ and $|P_2|$ are the number of patterns of two data sets used for the pattern classification. The ACA with highest efficiency  and {\em attrset-1} and {\em attrset-2} for the ACA are output of Algorithm \ref{algocla}. 

A list of candidate ACAs with their efficiencies are reported as an example of training in Table \ref{perCA}. Monk-1 data set is considered for this training. It has been obvious that rule 4 ACA is the classifier for Monk-1 data set (Table \ref{perCA}). Since the cells of an ACA are updated randomly, and an ACA state may converge to different fixed point attractors (Corollary-\ref{cor}), so the efficiency of an ACA may vary for different run of the classifier. Hence, the efficiencies shown in Table \ref{perCA} may change for different run of the classifier algorithm (Algorithm \ref{algocla}). To understand the effectiveness of the desired classifier, we need to test the classifier with a new set of patterns. This phase is commonly known as testing phase.


\begin{algorithm}
 \caption{Training Classifier}
 \begin{algorithmic}[1]
 \label{algocla}
 \renewcommand{\algorithmicrequire}{\textbf{Input:}}
 \renewcommand{\algorithmicensure}{\textbf{Output:}}
 \REQUIRE Table \ref{FLP}, $n$ (Size of ACA), Two pattern sets $P_1$ and $P_2$.
 \ENSURE  ACA as classifier, {\em attrset-1} and {\em attrset-2}.
 \\ \textit{Initialization} : Success= 0
  \STATE For each ACA, $A$  $\in$ Table \ref{FLP} repeat {\em Step 2} to {\em Step 14}.
  \STATE Repeat {\em Step 3} and {\em Step 4} for each pattern $p$ of $P_1$ and $P_2$.
   \STATE Load $A$ with  $p$.
    \STATE Run $A$ until it reaches to any fixed point attractor, $attr$.
     \STATE Suppose $n_1$ and $n_2$ are the number of patterns from $P_1$ and $P_2$ respectively mapped into an fixed point attractor, $attr$. \\
  \IF{ $n_1$ $\textgreater$ ~$n_2$}
  \STATE Success = Success + $n_1$ and store $attr$ in {\em attrset-1}.
 \ELSE 
 \STATE Success = Success + $n_2$ and store $attr$ in {\em attrset-2}.
    \ENDIF
      \STATE  Repeat {\em Step 5} to {\em Step 10} for each fixed point.
       \STATE Find efficiency as $\frac{Success}{|P_1| + |P_2|}$ 
        \STATE Report maximum efficiency ACA and its class of attractors.
 \end{algorithmic}
 \end{algorithm}

\begin{table}
\caption{Efficiencies of ACAs during training of monk-1 data-set }   
\label{perCA} 
\centering
\scalebox{0.79}{      
\begin{tabular}{|c|c|c|c|c|c|c|c|c|}\hline   
ACAs & Efficiency & number of & ACAs & Efficiency & number of & ACAs & Efficiency & number of \\  
  & (in \%)&attractors &  &(in \%)&attractors& &(in \%)&attractors\\ \hline  
4& 95.671 & 199 & 5 & 73.38 & 22 & 12 & 95.16 & 199 \\ \hline  
13& 73.38 & 22 & 36 & 85.48 & 67 & 44 & 84.67 & 67 \\ \hline  
68& 94.35 & 199 & 69 & 69.35 & 22 & 72 & 73.38 & 67 \\ \hline  
77 & 83.06 & 198 & 78 & 70.96 & 23 & 79 & 68.54 & 22 \\ \hline
92 & 73.38 & 23 & 93 & 73.38 & 22 & 94 & 71.77 & 23 \\ \hline  
95 & 70.96 & 22 & 100 & 81.45 & 67 & 104 & 70.16 & 34 \\ \hline  
128 & 50.00 & 2 & 130 & 50.00 & 2 & 132 & 89.51 & 200  \\ \hline  
133 & 71.77 & 23 & 136 & 50.00 & 2 & 138 & 50.00 & 2 \\ \hline  
141 & 72.58 & 23 & 144 & 50.00 & 2 &146 & 50.00 & 2 \\ \hline
152 & 50.00 & 2 & 160 & 52.41 & 2 & 162 & 56.45 & 2 \\ \hline   
164 & 82.25 & 68 & 168 & 52.41 & 2 & 170 & 60.48 & 2 \\ \hline  
172 & 83.87 & 68 & 174 & 50.00 & 2 & 176 & 53.22 & 2 \\ \hline  
178 & 62.09 & 2 & 182 & 50.00 & 2 & 184 & 62.60 & 2  \\ \hline  
186 & 55.64 & 2 & 188 & 50.00 & 2 & 190 & 50.00 & 2  \\ \hline  
192 & 50.00 & 2 & 194 & 50.00 & 2 & 197 & 69.35 & 23 \\ \hline 
202 & 75.00 & 68 & 203 & 77.41 & 67 & 207 & 91.93 & 199 \\ \hline  
208 & 50.00 & 2 & 216 & 77.41 & 68 & 217 & 86.29 & 67   \\ \hline  
218 & 83.06 & 68 & 219 & 83.06 & 67 & 221 & 91.93 & 199  \\ \hline  
222 & 87.90 & 200 & 223 & 93.54 & 199 & 224 & 52.41 & 2 \\ \hline  
226& 58.870 & 2& 228 & 81.45 & 68 & 232 & 85.48 & 200 \\ \hline   
234 & 51.61 & 2 & 237 & 71.77 & 67 & 238 & 50.00 & 2 \\ \hline  
240 & 59.67 & 2& 242 & 53.22 & 2 & 244 & 50.00 & 2 \\ \hline  
246 & 50.00 & 2 & 248 & 51.61 & 2 & 250 & 51.61 & 2  \\ \hline  
252 & 50.00 & 2& 254 & 50.00 & 2 &- &- & -  \\ \hline  
\end{tabular} }     
\end{table}  

\subsection{Testing Phase} 
In the {\em testing phase} the proposed classifier is tested with a new set of patterns. We use Algorithm \ref{algoclatest} to test the designed classifier. The classifier, two pattern sets ($P_1$ and $P_2$) and the class of attractors ({\em attrset-1} and {\em attrset-2}) are the input to Algorithm \ref{algoclatest}. The ACA is loaded with the patterns of $P_1$ and $P_2$ and updated till all the patterns converge to any fixed point attractor. The efficiency of the ACA is determined considering the number of patterns correctly identified by the classifier. For example, if an attractor is present in {\em attrset-1} then the algorithm counts only the number of patterns from data set $P_1$ converge to the attractor as correctly identified patterns. The ACAs with their training and testing efficiencies for different data sets are reported in Table \ref{perdata}. 

The proposed classifier is an ACA based classifier, so the cells of ACA are updated stochastically. So, there is a chance that the efficiency of the classifier can vary in different run. So, this variations in efficiencies of the classifier need to be reported. To overcome this problem of variation in efficiencies, we find the {\em margin of error} in both {\em training} and {\em testing} phase. 

\begin{algorithm}
 \caption{Testing Classifier}
 \begin{algorithmic}[1]
 \label{algoclatest}
 \renewcommand{\algorithmicrequire}{\textbf{Input:}}
 \renewcommand{\algorithmicensure}{\textbf{Output:}}
 \REQUIRE ACA, {\em attrset-1}, {\em attrset-2}, pattern sets $P_1$ and $P_2$, 
 \ENSURE Testing efficiency.
 \\ \textit{Initialization} : Success = 0.
  \STATE Repeat {\em Step 2} to {\em Step 4} for each pattern $p$ of $P_1$ and $P_2$.
   \STATE Load ACA with  $p$.
    \STATE Run ACA until it reaches to any fixed point attractor, $attr$.
     \STATE Suppose $n_1$ and $n_2$ are the number of patterns from $P_1$ and $P_2$ respectively mapped into an fixed point attractor, $attr$. 
     \IF{ $attr$ $\in$ {\em attrset-1}} 
      \STATE Success=Success + $n_1$    
        \ELSIF {$attr$ $\in$ {\em attrset-2}} 
         \STATE     Success=Success + $n_2$ 
          \ELSIF { ($attr$ $\notin$ {\em attrset-1}) $\&$ ($attr$ $\notin$ {\em attrset-2})}     
                \IF {$n_1$ $>$ $n_2$} 
                \STATE Success = Success + $n_1$ 
               \ELSE 
                \STATE Success = Success + $n_2$
                \ENDIF 
     \ENDIF 
     \STATE Repeat {\em Step 4} to {\em Step 15} for each fixed point attractor.
       \STATE Efficiency = $\frac{Success}{|P_1| + |P_2|}$ 
        \STATE Report Testing efficiency of the ACA.
 \end{algorithmic}
 \end{algorithm}

\subsection{Margin of Error} 

A {\em margin of error} expresses the maximum expected difference between the true population parameter and a sample estimate of that parameter \cite{William}.  We estimate the {\em margin of error} for sample size $m$ using Equation \ref{mr} \cite{William}. We have considered 30 samples for the experimentation.

\begin{equation}
\label{mr}
M\_ Error=Z_{n/2}~(\frac{S}{\sqrt m}) 
\end{equation}
where S is the variance and S can be found from Equation \ref{var}.

\begin{equation}
\label{var} 
S=\sqrt{\sum(x_i-\overline x)^2/(m-1)}
\end{equation}
Where $x_i$ is the efficiency of $i^{th}$ sample and $\overline x$ is the mean of efficiencies of samples. We set $Z_{n/2}$= 1.96, as we consider the confidence level is 95\% for our sampling experimentation \cite{William}.

The {\em margin of error} in both {\em training} and {\em testing} phase for the efficiencies of the proposed classifier are reported in Table \ref{perdata}. It is found that the {\em margin of error} is very less in both {\em training} and {\em testing} phase. Hence, the efficiencies of the proposed classifier varies a little during different run of the classifier algorithm.

\subsection{The Comparison}
We have used six data-sets for the study of effectiveness of the proposed pattern classifier. These data sets are Monk-1, Monk-2, Monk-3, Haber-man, Spect heart and Tic-tac-toe. These data-sets are modified suitably and fit with the input characteristics of our proposed pattern classifier. Our proposed classifier is a two-class classifier. So, all the data-sets we have used are of two classes. 

Table \ref{perCA} shows the efficiencies and the number of attractors of the candidate ACAs. The efficiencies reported in this table are found using Monk-1 data set during training phase. It is also noted that, the efficiency of classifier (in training) changes if data-set changes. 

Table \ref{perdata} shows the performance results of our proposed pattern classifier. The efficiencies for the proposed classifiers using different data sets in both {\em training} and {\em testing} phase with their {\em margin of errors} are also noted in this table. Column 1 shows the name of data-set while column 2 reports the size of ACA. The efficiencies of the classifier with their {\em margin of errors} during training are reported in next two columns and efficiencies during testing with {\em margin of errors} are reported in column 5 and column 6 respectively. The ACA, as the proposed classifier is reported in the last column. Finally, the performance of our proposed ACA based classifier is compared with other well known classifiers in Table  \ref{compdata}. We show that our proposed ACA based classifier performs much better than traditional CA based classifier and also more competitive and performs reliably better than other well known classifier algorithms. 

\begin{table*}
\caption{Performance of proposed classifier}   
\label{perdata} 
\centering 
\scalebox{0.8}{   
\begin{tabular}{|c|c|c|c|c|c|c|}\hline   
Data-sets & ACA size& \multicolumn{4}{|c|}{Efficiency in \%} & ACA rules  \\  \cline{3-6}
 &  & Training& Margin of error & Testing&Margin of error & (Proposed) \\  
 &    &         & in Training &  &  in Testing& \\ \hline
Monk-1 & 11 & 95.671 & 0.482 & 81.519 &0.254 &4  \\ \hline       
Monk-2 & 11 & 88.934 & 0.352 & 73.410 &0.305 &68  \\ \hline 
Monk-3 & 11 & 96.310 & 0.485 & 83.749& 0.134 &207 \\ \hline   
Haber Man & 9 & 82.179 & 0.405 & 77.493& 0.773 &36 \\ \hline  
Spect Heart & 22 & 100.000 & 0 & 100.000 & 0 & 4 \\ \hline   
Tic-Tac-Toe & 18 & 99.867 & 0.059 & 99.721& 0.040 &12 \\ \hline   
\end{tabular} }    
\end{table*}  

\begin{table}
\caption{Comparison of classification accuracy}   
\label{compdata} 
\centering 
\scalebox{1.0}{  
\begin{tabular}{|c|c|c|c|}\hline   
Data-sets & Algorithms &  Efficiency      
 & Efficiency  \\  
 &  &in \% &in \% (proposed) \\ 
 &   &    &  with margin of error \\ \hline
Monk 1&Bayesian & 99.9& 81.519 $\pm$ 0.254 \\
      &C4.5&100& (rule 4)\\ 
      &TCC&100&\\
      &MTSC&98.65& \\
      &MLP&100&\\
      &Traditional CA&61.111&\\
     \hline
Monk 2 & Bayesian & 69.4& 73.410 $\pm$ 0.305 \\ 
      &C4.5&66.2&(rule 68)\\ 
      &TCC&78.16&\\
      &MTSC&77.32& \\
      &MLP&75.16&\\
      &Traditional CA&67.129&\\
      \hline
Monk 3 &Bayesian & 92.12& 83.749 $\pm$ 0.134 \\ 
      &C4.5&96.3&(rule 207)\\ 
      &TCC&76.58&\\
      &MTSC&97.17& \\
      &MLP&98.10&\\ 
      &Traditional CA&80.645&\\
     \hline
Haber-man &Traditional CA & 73.499 & 77.493 $\pm$ 0.773\\ 
& & & (rule 36) \\ \hline       
Spect Heart&Traditional CA &91.978 &100 $\pm$ 0.0\\ 
 & & & (rule 4) \\ \hline
 
Tic-Tac-Toe&Sparse grid&98.33&99.721 $\pm$ 0.040\\ 
      &ASVM&70.000&(rule 12) \\
      &LSVM&93.330& \\
      &Traditional CA&63.159&\\
      \hline  
\end{tabular}}  
\end{table}  

\section{Conclusion}
\label{conclu}
In this work we have proposed a design of ACA based pattern classifier and compared the efficiency of the designed classifier with other well known existing classifiers. For this design, we have used ACAs under fully asynchronous update scheme in periodic boundary condition. ACAs are characterized for their convergence towards any fixed point attractors. A theorem has been designed to identify such convergent ACAs which are used for the pattern classification. 146 ACAs (out of 256) are identified, that converge to some fixed point attractors during their evolution. The concept of {\em fixed point graph} has been introduced to facilitate the counting of fixed points in an ACA. An algorithm is also designed which counts the number of fixed points utilizing FPG and identifies ACAs with multiple attractors. 84 ACAs (out of 146) are identified with multiple fixed points. An experimental study is made on the convergence time of ACAs with a designed algorithm. All convergent ACAs (146 ACAs) with their convergence time are reported in this work. ACAs with exponential convergence time and ACAs with huge number of fixed points are not allowed as pattern classifier. These ACAs are not good candidates for pattern classification. 71 ACAs are identified as candidate ACAs which are used as pattern classifier. Both training and testing of the classifier have been done considering widely accepted data sets. Since there is a chance of change in the efficiency of the ACA based classifier, we have also calculated the margin of error in both {\em training} and {\em testing} phase. It has been shown that the margin of errors in both the phases are very less with a negligible effect in the efficiency of the classifier. Finally, in this work we have also compared the efficiency of our proposed ACA based classifier with some well known existing classifiers. We report that our proposed classifier performs better than the traditional CA based classifier and also more competitive and performs reliably better than some of the other well known classifiers.   
\bibliographystyle{elsarticle-num}
\bibliography{References} 

\begin{thebibliography}{10}
\expandafter\ifx\csname url\endcsname\relax
  \def\url#1{\texttt{#1}}\fi
\expandafter\ifx\csname urlprefix\endcsname\relax\def\urlprefix{URL }\fi
\expandafter\ifx\csname href\endcsname\relax
  \def\href#1#2{#2} \def\path#1{#1}\fi

\bibitem{Neuma66}
J.~V. Neumann, The theory of self-reproducing Automata, {A}. {W}. {B}urks ed.,
  Univ. of {I}llinois {P}ress, {U}rbana and {L}ondon, 1966.

\bibitem{Wolframbook1}
S.~Wolfram, A {N}ew {K}ind of {S}cience, Wolfram Media, Inc., 2002.

\bibitem{LangtonIII}
C.~G. Langton, Self-reproduction in cellular automata, Physica D 10 (1984)
  135--144.

\bibitem{langton86}
C.~G. Langton, Studying artificial life with cellular automata, Physica D 22
  (1986) 120--149.

\bibitem{Burks70}
A.~W. Burks, Essays on cellular automata, Tech. rep., Univ. of Illinois, Urbana
  (1970).

\bibitem{chris2004}
C.~Salzberg, H.~Sayama, Heredity, {C}omplexity, and {S}urprise: Embedded
  self-replication and evolution in {CA}, in: Proceedings of International
  conference on cellular automata for research and industry (ACRI), Springer,
  2004, pp. 161--171.

\bibitem{reggia1998}
J.~A. Reggia, H.~H. Chou, J.~D. Lohn, Cellular automata models of
  self-replicating systems, Advances in computers 47 (1998) 141--183.

\bibitem{Chopard}
B.~Chopard, M.~Droz, {C}ellular {A}utomata {M}odelling of {P}hysical {S}ystems,
  {C}ambridge {U}niversity {P}ress, 1998.

\bibitem{Codd68}
E.~F. Codd, Cellular {A}utomata, Academic Press Inc., 1968.

\bibitem{palash3}
P.~Sarkar, A brief histroy of cellular automata, Acm Computing Surveys 32~(1)
  (2000) 80--107.

\bibitem{Maxgarzon}
M.~Garzon, Models of massive parallelism: Analysis of cellular automata and
  neural networks, Springer, 1995.

\bibitem{mitchell93}
M.~Mitchell, P.~T. Hraber, J.~P. Crutchfield, Revisiting the egde of chaos:
  Evolving cellular automata to perform computations, Complex Systems (1993) 7:
  89--130.

\bibitem{toffoli90}
T.~Toffoli, N.~Margolus, Invertible cellular automata: {A} review, Physica D 45
  (1990) 229.

\bibitem{Thatcher}
J.~Thatcher, Universality in {V}on {N}eumann cellular model, in: Tech. Report
  03105-30-T, ORA, University of Michigan, 1964.

\bibitem{wolfram84b}
S.~Wolfram, Universality and complexity in cellular automata, Physica D 10
  (1984) 1--35.

\bibitem{jarkko2005}
J.~Kari, Theory of cellular automata: A survey, Theoretical computer science
  (Elsevier) 334 (2005) 3--33.

\bibitem{gutowitz91}
H.~Gutowitz, Cellular Automata: Theory and Experiment, MIT Press/Bradford
  Books, Cambridge Mass., 1991, {I}SBN 0-262-57086-6.

\bibitem{cook2004}
M.~Cook, Universality in elementary cellular automata, Complex systems 15
  (2004) 1--40.

\bibitem{Wolfr83}
S.~Wolfram, Statistical mechanics of cellular automata, Rev. Mod. Phys. 55~(3)
  (1983) 601--644.

\bibitem{wolfram86}
S.~Wolfram, Theory and applications of cellular automata, World Scientific,
  Singapore, 1986, {I}SBN 9971-50-124-4 pbk.

\bibitem{wolfram1984}
S.~Wolfram, Cellular automata as models of complexity, Nature 311~(5985) (1984)
  419--424.

\bibitem{gianluca1995}
G.~Tempesti, A new self-reproducing cellular automaton capable of construction
  and computation, in: Proceedings of 3rd European conference on artificial
  life, no. 929, Springer, 1995.

\bibitem{gorodkin1993}
J.~Gorodkin, A.~Sorensen, O.~Winther, Neural network and cellular automata
  complexity, Complex systems 7 (1993) 1--23.

\bibitem{hoekstra2010}
A.~G. Hoekstra, J.~Kroc, P.~M.~A. Sloot, Introduction to modeling of complex
  systems using cellular automata, springer, 2010.
\newblock \href {http://dx.doi.org/10.1007/978-3-642-12203-3-1}
  {\path{doi:10.1007/978-3-642-12203-3-1}}.

\bibitem{Meinhardt1995}
H.~Meinhardt, The Algorithmic beauty of sea shells, springer, 1995.

\bibitem{patel2015}
E.~L. Patel, D.~Broomhead, A max-plus model of asynchronous cellular automata,
  arXiv:1502.04097 (2015) 1--26.

\bibitem{Fates14}
N.~Fat{\`{e}}s, Guided tour of asynchronous cellular automata, J. Cellular
  Automata 9~(5-6) (2014) 387--416.

\bibitem{damien2009}
D.~Regnault, N.~Schabanel, E.~Thierry, Progresses in the analysis of stochastic
  2{D} cellular automata: A study of asynchronous 2{D} minority, Theoretical
  Computer Science (Elsevier) 410 (2009) 4844--4855.

\bibitem{bersinietal94}
H.~Bersini, V.~Detour, Asynchrony induces stability in cellular automata based
  models, in: R.~A. Brooks, P.~Maes (Eds.), Artificial Life {IV}, The MIT
  Press, Cambridge, Massachusetts, 1994, pp. 382--387.

\bibitem{Inger84}
T.~Ingerson, R.~Buvel, Structure in asynchronous cellular automata, Physica D:
  Nonlinear Phenomena 10~(1-2) (1984) 59--68.

\bibitem{Nehaniv2003}
C.L.Nehaniv, Evolution in asynchronous cellular automata, in: Proceedings of
  eighth international conference on artificial life, MIT Press, 2003, pp.
  65--73.

\bibitem{Maji1}
P.~Maji, N.~Ganguly, S.~Saha, A.~K. Roy, P.~P. Chaudhuri, {C}ellular {A}utomata
  {M}achine for {P}attern {R}ecognition, {P}roceedings of {F}ifth
  {I}nternational {C}onference on {C}ellular {A}utomata for {R}esearch and
  {I}ndustry, {ACRI}, {S}witzerland (2002) 270--281.

\bibitem{Huang5}
F.~Peper, T.~Isokawa, N.~Kouda, N.~Matsui, Self timed cellular automata and
  their computational ability, Future generation computer systems (elsevier) 18
  (2002) 893--904.

\bibitem{Huang4}
X.~Huang, Q.~Zhu, Self reproduction of worms in asynchronous cellular automata,
  Journal of software (Academy publishing) 8 (2013) 1699--1706.

\bibitem{Huang3}
X.~H.~Q. Zhu, A novel triggered self reproduction of self reproduction loops in
  asynchronous cellular automata, Journal of information and computational
  science 9~(16) (2012) 4961--4968.

\bibitem{Huang1}
X.~Huang, J.~Lee, R.~L. Yang, Q.~S. Zhu, Simple and flexible self-reproducing
  structures in asynchronous cellular automata and their dynamics,
  International Journal of modern physics C 24~(1350015).

\bibitem{Ytakada}
Y.~Takada, T.~Isokawa, F.~Peper, Asynchronous self-reproducing loops with
  arbitration capability, Physica D: Nonlinear phenomena 227 (2007) 26--35.

\bibitem{Bolt2015}
W.~Bolt, J.~M. Baetens, B.~D. Baets, On the decomposition of stochastic
  cellular automata, arXiv:1503.03318 (2015) 1--27.

\bibitem{DasMNS09}
S.~Das, S.~Mukherjee, N.~Naskar, B.~K. Sikdar, Characterization of single cycle
  ca and its application in pattern classification, Electr. Notes Theor.
  Comput. Sci. 252 (2009) 181--203.

\bibitem{das:modeling}
S.~Das, S.~Mukherjee, N.~Naskar, B.~K. Sikdar, Modeling single length cycle
  nonlinear cellular automata for pattern recognition, in: NaBIC, 2009, pp.
  198--203.

\bibitem{NiloyPAMI}
N.~Ganguly, P.~Maji, B.~K. Sikdar, P.~P. Chaudhuri, {D}esign of a {C}ellular
  {A}utomata {B}ased {P}attern {C}lassifier, {T}ransaction on {P}attern
  {A}nalysis and {M}achine {I}ntelligence, {TPAMI}~(116429).

\bibitem{NiloyI}
N.~Ganguly, P.~Maji, A.~Das, B.~K. Sikdar, P.~P. Chaudhuri, {C}haracterization
  of {N}on-{L}inear {C}ellular {A}utomata {M}odel for {P}attern {R}ecognition,
  in: {P}roceedings of {AFSS} {I}nternational {C}onference on {F}uzzy
  {S}ystems, {C}alcutta, {I}ndia, 2002, pp. 214--220.

\bibitem{hpcs}
B.~Sethi, S.~Das, Modeling of asynchronous cellular automata with fixed-point
  attractors for pattern classification, in: Proceedings of International
  Conference on High Performance Computing and Simulation, Finland, IEEE, 2013,
  pp. 311--317.

\bibitem{nazim}
N.~Fat{\`e}s, E.~Thierry, M.~Morvan, N.~Schabanel, Fully asynchronous behavior
  of double-quiescent elementary cellular automata, Theor. Comput. Sci.
  362~(1-3) (2006) 1--16.

\bibitem{tamc2014}
B.~Sethi, N.~Fat{\`e}s, S.~Das, Reversibility of elementary cellular automata
  under fully asynchronous update, in: Proceedings of International conference
  on Theory and Applications of Models of Computation, India, Springer, 2014,
  pp. 39--49.

\bibitem{automata2013}
B.~Sethi, S.~Roy, S.~Das, Experimental study on convergence time of elementary
  cellular automata under asynchronous update, in: Proceedings of International
  Workshop on Cellular Automata and Discrete Complex Systems (Automata), 2013,
  pp. 223--228.

\bibitem{William}
W.~G. Cochran, Sampling Techniques, third edition, Vol. ISBN 978-81-265-1524-0,
  John Wiley \& Sons, 1977.

\bibitem{Catherine}
C.~McGeoch, P.~Sanders, R.~Fleischer, P.~R.Cohen, D.~Precup, Using finite
  experiments to study asymptotic performance, Experimental Algorithmics,
  LNCS~(2547) (2002) 93--126.

\end{thebibliography}
\end{document}